\documentclass[11pt]{article}
\usepackage[utf8]{inputenc}

\title{Ranking with Partitioning}
\author{Samuel Boardman \\ 
        University of Michigan \\ 
        \texttt{stboard@umich.edu}}

\usepackage[letterpaper, margin=1.0in]{geometry}
\usepackage{hyperref}
\usepackage{amsmath}
\usepackage{amssymb}
\usepackage{amsthm}
\newtheorem{theorem}{Theorem}[section]
\newtheorem{corollary}[theorem]{Corollary}
\newtheorem{lemma}[theorem]{Lemma}
\newtheorem{definition}[theorem]{Definition}

\usepackage{algorithm}
\usepackage{algpseudocode}
\usepackage{graphicx}
\usepackage{tikz}
\usetikzlibrary{positioning,chains,fit,shapes,calc}
\usetikzlibrary{decorations.pathreplacing}
\usepackage{adjustbox}
\usepackage{pdflscape}
\usepackage[edges]{forest}
\usepackage{dirtree}
\usepackage{array}
\usepackage[T1]{fontenc}
\usepackage{caption}
\usepackage{subcaption}
\usepackage{float}
\usepackage[style=numeric,sorting=none]{biblatex}
\newlength{\treeindent}
\newcommand{\treeitem}[2]{{\leftskip=#1\treeindent \noindent #2\par}}

\makeatletter
\newcommand*\widebar[1]{%
  \hbox{%
    \vbox{%
      \hrule height 0.5pt 
      \kern0.5ex
      \hbox{%
        \kern-0.1em
        \ensuremath{#1}%
        \kern-0.1em
      }%
    }%
  }%
}
\makeatother

\begin{document}

\maketitle

\begin{abstract} \noindent
Given an undirected graph representing similarities between a set of items and an additive measure evaluating the items, we treat the position of a special subset of items in an ordinal ranking through a collection of combinatorial optimization problems in which items may be combined if they are similar. The objective for these problems is to either maximize or minimize the absolute or relative rank of the special subset, with a meta-goal of assessing the robustness of the rank, even in the presence of a well-defined criterion. We classify the computational complexity of all four problems, mostly finding worst-case hardness, then find exact and approximate solutions to special cases and variants of the problems. These structured cases are inspired by several real-world examples and may be used to assess commonly cited facts across disparate domains, as we demonstrate for sources of greenhouse gas emissions that contribute to climate change. 
\end{abstract}

\section{Introduction}

Rankings play a central role in how information is summarized and communicated in the Information Age. Online platforms routinely rank cities by livability, universities by academic quality, and pieces of media by popularity or quality. Outside of explicitly digital settings, rankings induce metrics that are frequently cited as evidence in policy debates, journalism, and scientific communication. \\

Despite generally being presented as plain facts in the latter contexts, many such rankings are constructed from entities that admit multiple reasonable levels of aggregation. Cities may be grouped into metropolitan areas or countries, food ingredients into the byproducts formed when following a recipe, and television episodes into seasons or series. These aggregations are typically justified by some notion of similarity, yet allow for many degrees of freedom in general, even under domain-specific constraints. As a consequence, the position of a particular item or collection of items in a ranking may depend not only on the underlying measure being applied, but also on how similar items are combined or kept separate. \\

This observation raises a basic robustness question: to what extent is the rank of a given item intrinsic to the measure being used, and to what extent can it vary under different groupings by similar items? \\

In the present paper, we formalize this question using a framework in which items are represented as vertices of a graph encoding similarity, admissible aggregations correspond to connected subsets of vertices, and rankings arise from additive measures evaluated on these subsets. Within this framework, we study optimization problems that minimize or maximize the rank, or a normalized version called rank percentile, of a designated subset, thereby quantifying the range of rankings consistent with a fixed underlying metric. Even when historical and legal considerations preclude any meaningful possibility of regrouping, say, land into cities, such considerations are often the result of arbitrary but solidified decisions. Studying all possible partitions enables us to assess counterfactual developments, providing one angle to decide which sociological facts are largely a matter of the definitions and boundaries chosen. \\

This paper is organized as follows. Section \ref{general_case} presents the general formulation of the optimization problems, formalizing our motivating examples involving cities and episodes. As three of the four problems stated are {\sc NP}-hard, Section \ref{special_cases} introduces simplifications to our model that arise naturally when modeling several partitioning contexts, including the aforementioned episode example, to make these problems more tractable, yielding polynomial-time optimal solutions and approximation algorithms. Section \ref{variants} adjusts the set-up of our problems as needed to describe related aggregation problems that may not be realized in the framework of Section \ref{general_case}, particularly in the presence of hierarchical categories as in ingredient lists. Our results in Section \ref{variants} also facilitate the application of our methods to study different sources of greenhouse gas emissions in the United States in Section \ref{numerical_experiments}, revealing that the relative contribution of many sources depends heavily on our choice of aggregation, even though these ordinal ranking are frequently cited without elaboration. Section \ref{future_work} anticipates future directions that would make our theoretical results and real-world applications more comprehensive. \\ 

While this paper focuses on the new combinatorial optimization problems motivated and introduced above, they enjoy several special cases and variants that are equivalent to well-studied set packing and graph problems, enabling us to better understand the computational complexity and (in)approximability of these alterations. Maximizing the number of vertices selected from intersection and overlap graphs, particularly rectangle intersection graphs, has been shown to be computationally challenging; early hardness results were established by Rim and Nakajima \cite{mdshardness}, and more recent work has extended these limitations to closely related graph classes such as CPG and EPG graphs \cite{integralrectanglehardness}. Despite these barriers, approximation algorithms have been developed for related problems, including the $(2+\epsilon)$-approximation for maximum independent set of rectangles due to Gálvez et al. \cite{mdsapprox}. The $k$-set packing problem, which arises as an extremal case of partitioning in the present paper, is also hard, but has been approached via large-neighborhood local search techniques by Sviridenko and Ward \cite{localsearchapprox}. Structural graph theory further illuminates settings where structure simplifies our problems, as in the characterization of Hamiltonian properties in abelian group graphs by Chen and Quimpo \cite{hampathcirculant}.

\section{The General Case of the Problem} \label{general_case}

Consider $n$ indivisible elements $V = \{v_1,\dots,v_n\}$ that have some interpretation as atomic entities in some category, such as a city or an episode of a show. There is a binary relation $E \subset V \times V$ representing whether two entities are similar and may be combined, as in the cities forming a land mass (e.g., a country) or a series of consecutive episodes of a show (e.g., a season). The relation $E$ is reflexive and symmetric. As suggested by this naming, we may identify such a tuple $(V,E)$ with an undirected (simple) graph $G = (V,E)$, using context to distinguish between the interpretations. \\ 

Now, we may partition $V$ into subsets $S_1,\dots,S_c$ taken from a collection of subsets $\mathcal{S} \subset 2^V$. The collection $\mathcal{S}$ is an inductively defined set: each $S \in \mathcal{S}$ is either a singleton or the union of two subsets in $\mathcal{S}$ forming a pair of similar elements (and $\mathcal{S}$ comprises all such subsets). In other words, $\mathcal{S}$ consists of the subsets $S$ that can be built up by starting with an initial element and successively adding elements similar to some preexisting element (note that the initial element may be taken arbitrarily without loss of generality, since $E$ is symmetric and we can backtrack as needed). To see the equivalence, clearly any such subset satisfies the inductive definition, and when taking the union of two sets $S_1 \cup S_2 \in \mathcal{S}$ as above, if they both satisfy this property (tautological for singletons), then we can start with any element in the union, successively add similar elements from that subset, and, after encountering the similar element from the other subset, add the similar elements from that subset. Thus by structural induction, the property holds for $S_1 \cup S_2$. In graph theoretic terms, $\mathcal{S}$ consists of the subsets $S$ that constitute a walk in $G$ (because two elements being similar means there exists an edge between them in $G$, and we may use backtracking to toggle between the preexisting elements to only include the ones we want). \\ 

Equivalently, $S \in \mathcal{S}$ if and only if for every partition $S = A \sqcup B$, $A,B \neq \emptyset$, there exists $a \in A$ and $b \in B$ such that $(a,b) \in E$ (we say such a subset is \textit{connected} in analogy to the topological definition). If $S = \{u_1,\dots,u_k\}$ is connected, we can successively build $S \in \mathcal{S}$ starting from $u_1$ and at each step adding an element of $S$ that has not yet been added (apply the definition of connectivity to the preexisting set and its complement within $S$). On the other hand, if $S \in \mathcal{S}$ and $S$ is not connected, this means two subsets $A$ and $B$ forming $S$ are in different connected components of $G$, which contradicts that there be a walk in $G$ consisting of the vertices of $S$. \\

Lastly, let $\mu: 2^V \to [0,\infty)$ be a finite measure given by $\mu(S) = \sum_{s \in S} \mu(\{s\})$, where $\mu(S) = 0 \implies S = \emptyset$. Let $S^* \in \mathcal{S} \cup \{\emptyset\}$ be a special subset whose value $\mu(\mathcal{S})$ we wish to compare to the values $\mu(S)$ for subsets $S \in \mathcal{S}$. Consequently, we write $G \backslash S^*$ to denote the graph with vertex set $V \backslash S^*$ and edge set $E$ minus any edges incident to an element of $S^*$. \\

The {\sc Rank Minimization} problem is the following: given a tuple $(V,E,\mu,S^*)$ as above, select a partition $(S_1,\dots,S_c,S^*)$ of $V$ satisfying $S_1,\dots,S_c \in \mathcal{S}$ (call such a partition \textit{valid}) that minimizes the number of $S_i$ for which $\mu(S_i) > \mu(S^*)$ (we say one plus this number is the \textit{rank} of $S^*$ under a valid partition). \\ 

The {\sc Rank Maximization} problem is similar: given a tuple $(V,E,\mu,S^*)$ as above, select a partition $(S_1,\dots,S_c,S^*)$ of $V$ satisfying $S_1,\dots,S_c \in \mathcal{S}$ that maximizes the number of $S_i$ for which $\mu(S_i) \geq \mu(S^*)$ (we say one plus this number is the \textit{rank} of $S^*$ under a valid partition). \\

The definition of rank differs slightly depending on the problem, and any claims in the present paper about the minimum or maximum rank in some context are satisfied for both definitions simultaneously. To avoid manipulation of rank by such tie-breaking precedent and the number or average size of subsets, we also consider percentile versions of these optimization problems. \\ 

The {\sc Rank Percentile Minimization} problem is the following: given a tuple $(V,E,\mu,S^*)$ as above, select a partition $(S_1,\dots,S_c,S^*)$ of $V$ satisfying $S_1,\dots,S_c \in \mathcal{S}$ that minimizes the fraction $$\frac{|\{i \in [c]: \mu(S_i) > \mu(S^*)\}| + .5|\{i \in [c]: \mu(S_i) = \mu(S^*)\}| + .5}{c+1}.$$ The {\sc Rank Percentile Maximization} problem is similar: given a tuple $(V,E,\mu,S^*)$ as above, select a partition $(S_1,\dots,S_c,S^*)$ of $V$ satisfying $S_1,\dots,S_c \in \mathcal{S}$ that maximizes the fraction $$\frac{|\{i \in [c]: \mu(S_i) > \mu(S^*)\}| + .5|\{i \in [c]: \mu(S_i) = \mu(S^*)\}| + .5}{c+1}.$$ In either case, we call this fraction the \textit{percentile} of $S^*$ under a valid partition. \\  

The following definitions will also be convenient for our analysis of the problems. 

\begin{definition}
We say that a subset $S \subset V$ is \textit{large} if $\mu(S) > \mu(S^*)$ and $S \in \mathcal{S} \backslash \{S^*\}$, \textit{medium} if $\mu(S) = \mu(S^*)$ and $S \in \mathcal{S} \backslash \{S^*\}$, and \textit{small} if $\mu(S) < \mu(S^*)$ and $S \in \mathcal{S} \backslash \{S^*\}$. By identifying $v \in V$ with $\{v\}$, we also say $v$ is large, medium, or small if $\{v\}$ is large, medium, or small, respectively.   
\end{definition}


  


As an example, we may consider the northeast region of the United States commonly referred to as New England, which is somewhat arbitrarily partitioned into states, at the more granular level of municipalities. Letting $V$ be equal to the set of municipalities in New England, $E$ the set of adjacent pairs of municipalities, and $\mu$ the population of an input subset of municipalities, the optimization problems allow us to fix just one pre-established contiguous region of interest to people, such as Boston or Massachusetts itself, and study the extent to which its relative population in the New England region could change under an alternate set of (still contiguous) cities or states. 

\begin{center}

    \includegraphics[width = .30\textwidth]{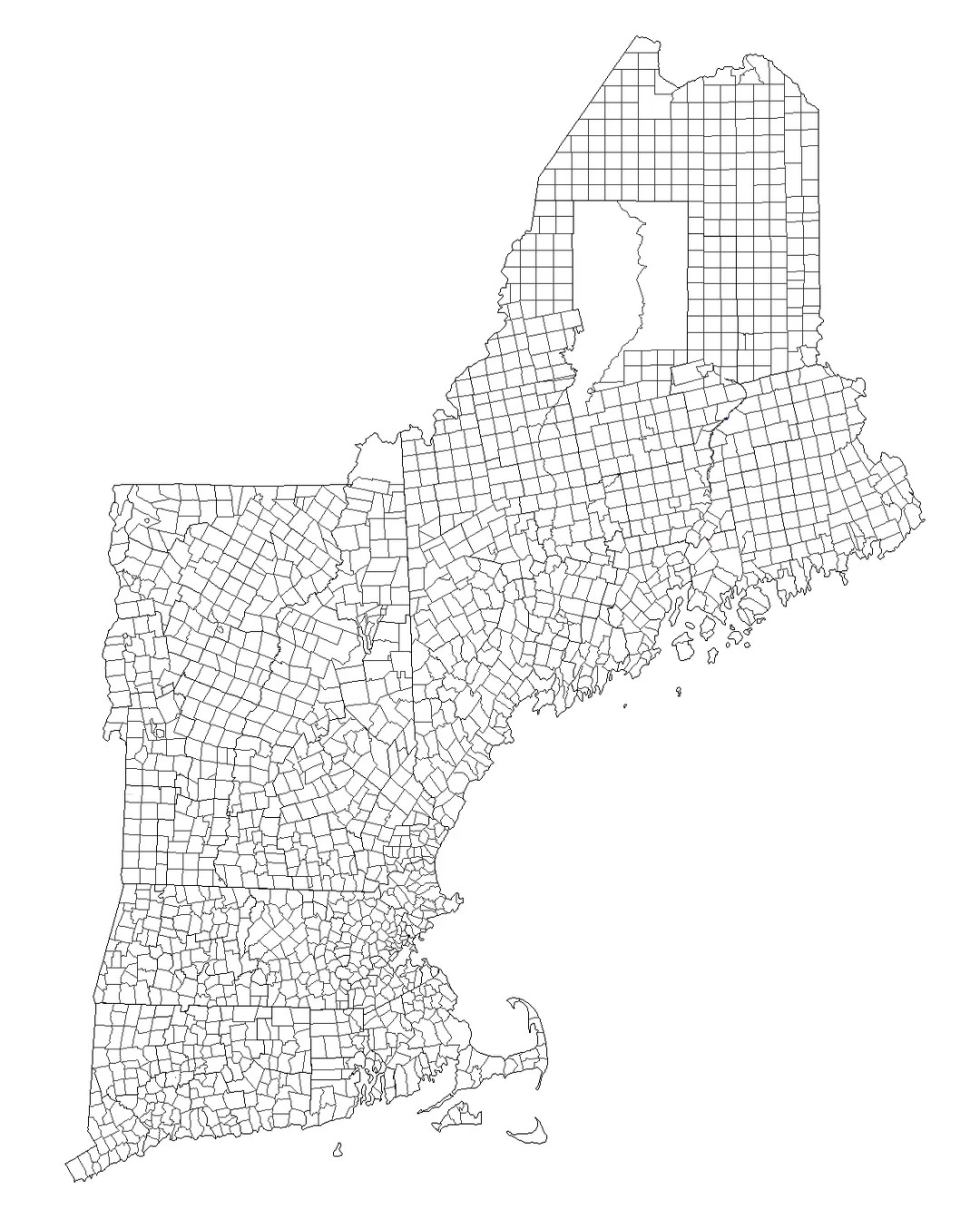}
    \captionof{figure}{A map depicting each municipality in New England by black boundary lines \cite{bryant2019newengland, NewEnglandmunicipalities}. There are many ways one could have partitioned these municipalities into contiguous states, each way carrying implications for how people may think about the relative standing of its constituent parts under some metric.}
    
\end{center}

To analyze the computational problems stated, we first recall the following standard arithmetic fact.

\begin{lemma} \label{frac_avg}
    Let $a,b,c,d \in \mathbb{R}$, where $b,d > 0$. Then $\min\left(\frac{a}{b},\frac{c}{d}\right) \leq \frac{a+c}{b+d} \leq \max\left(\frac{a}{b},\frac{c}{d}\right)$.
\end{lemma}

\begin{proof}
    Writing $\frac{a+c}{b+d} = \left(\frac{b}{b+d}\right)\frac{a}{b} + \left(\frac{d}{b+d}\right)\frac{c}{d}$, we see that $\frac{a+c}{b+d}$ is a convex combination of $\frac{a}{b}$ and $\frac{c}{d}$. 
\end{proof}

Our main application of Lemma \ref{frac_avg} is in understanding changes to the percentile of a partition upon combining subsets. 

\begin{lemma} \label{part_comb}
    Let $P_0$ be a valid partition of $V$. There exists a coarser valid partition achieving percentile less than (respectively greater than) that under $P_0$ if and only if there exist $l$ large subsets, $m$ medium subsets, and $s$ small subsets in $P_0$ whose union is in $\mathcal{S}$, where $l+m+s \geq 2$, such that $$\frac{l+.5m-1}{l+m+s-1}$$ is greater than (respectively less than) the percentile under $P_0$. 
\end{lemma}

\begin{proof}
    For the case $l = 0, m = 0, s \geq 2$, $\frac{l+.5m-1}{l+m+s-1} < 0$, which is necessarily less than the percentile of $P_0$; modifying $P_0$ so that the $s$ small subsets are combined into one subset, the denominator of the rank decreases whereas the numerator can only increase, so the percentile is greater than that of $P_0$. \\

    Otherwise, we know that $l \geq 1$ or $m \geq 1$. Therefore, upon combining $l$ large subsets, $m$ medium subsets, and $s$ small subsets into a single subset, the percentile becomes $\frac{l_0 - l + .5(m_0-m) + 1 + .5}{l_0 - l + m_0 - m + s_0 - s + 1 + 1}$. Now by Lemma \ref{frac_avg}, the percentile of $P_0$ (the initial percentile) is a convex combination of this new percentile and the fraction $\frac{l+.5m-1}{l+m+s-1}$, so the new percentile is less than the initial percentile if and only if $\frac{l+.5m-1}{l+m+s-1}$ is greater than the initial percentile, and the new percentile is greater than the initial percentile if and only if $\frac{l+.5m-1}{l+m+s-1}$ is less than the initial percentile. Since any coarser partition of $P_0$ may be realized by recursively combining subsets of this form, the claim follows by induction: if the new percentile has not decreased after a certain number of steps, then the condition for it to decrease (in terms of available subsets to combine and the required value of $\frac{l+.5m-1}{l+m+s-1}$) only becomes more stringent. The argument for the percentile increasing is symmetric.         
\end{proof}

The minimization problems are quite different from each other in that minimizing the rank tends to create relatively few subsets that are larger, which results in a high percentile.




\begin{theorem} \label{min_in_P}
    {\sc Rank Minimization} is polynomial-time solvable.
\end{theorem}

\begin{proof}
    Consider the following greedy algorithm. Set $i = 1$ and mark all nodes in $G \backslash S^*$ as unexplored. While there exists an unexplored $v_+ \in V \backslash S^*$ that is large (i.e., $\mu(\{v_+\}) > \mu(S^*)$), initialize $S_i = \emptyset$, perform breadth-first search in $G \backslash S^*$ starting at $v_+$. Add each node that becomes explored (traversed by the BFS subroutine) to $S_i$ and increment $i$. Denote by $V_-$ the set of nodes in $G \backslash S^*$ that remain unexplored upon the algorithm's termination. We claim that the subsets $S_1,\dots,S_c,\{v_-\}: v_- \in V_-$ generated by this procedure minimize the rank of $S^*$. By the above characterization in terms of walks in $G$ (which includes all walks in $G \backslash S^*$), they are each in $\mathcal{S}$, and by construction they partition $V \backslash S^*$. Observe that the rank of $S^*$ is one plus the number of connected components of $G \backslash S^*$ containing some $v_+ \in V_+$. But subsets spanning multiple connected components are not in $\mathcal{S}$ (nor intersect with $S^*$), so this is optimal.           
\end{proof}

This gives us our first lower bound on the rank of $S^*$. 

\begin{corollary} \label{rank_lower_bound}
    The minimum rank of $S^*$ under a valid partition is $1 + C_l$, where $C_l$ is the number of connected components with a singleton that is large.  
\end{corollary}

\begin{theorem} \label{inapprox}
    {\sc Rank Percentile Minimization} is {\sc NP}-hard.
\end{theorem}


\begin{proof}   

    We show that solving {\sc Rank Percentile Minimization}, or just determining whether any partition achieves a percentile less than the partition $P_0$ comprised of $S^*$ and singletons, solves exact cover with 3-sets, {\sc X3C}. We describe a reduction from an instance of {\sc X3C} with $n$ nodes and $m$ subsets to an instance of {\sc Rank Percentile Minimization}.   
    \begin{itemize}
        \item Create a node for each subset and element; the only edges are from subsets to the elements contained therein and between each pair of subset nodes
        \item The subset nodes are small and the element nodes are large (as singleton sets)
        \item Create isolated large and small nodes (interpreted as singleton sets) as needed so that the percentile of $S^*$ under $P_0$ is $p_0 := \frac{n-1}{n+n/3-1} - \frac{1}{2}(\frac{n-1}{n+n/3-1} - \frac{n-3-1}{n-3+n/3-1-1})$
        \item All small nodes are so small that the sum of their values is less than $\mu(S^*)$ 
    \end{itemize}
    We show that there is an exact set cover if and only if there exists a subset such that $\frac{l+.5m-1}{l+m+s-1} > p_0$. By Lemma \ref{part_comb}, this condition is equivalent to the statement that there exists a partition with percentile less than the percentile of $P_0$, so this completes the proof. \\
    
    $\implies$: Take the subset to consist of the subset nodes comprising an exact cover (and all of the element nodes), since the subset nodes are all connected and each element node is connected to the subset node covering it in the exact cover. Then $\frac{l+.5m-1}{l+m+s-1} = \frac{n-1}{n+n/3-1} > p_0$. \\
    
    $\impliedby$: We analyze such a subset and show its subset nodes constitute an exact cover. Isolated nodes cannot be part of the (non-singleton) subset. Similarly, any element node part of the subset must have a corresponding subset node covering it in the subset. For $m$ subset nodes in subset, want to add maximum $3m$ element nodes to the subset, but this is less than $p_0$ by construction and that each subset-3-element tuple added increases the fraction by Lemma \ref{frac_avg}. If instead $m = k$, at least one of the subset nodes only bring 2 new element nodes, but the fraction of $\frac{2}{3}$ decreases the fraction from the baseline given by a 1:3 ratio (true at minimum of the ``first'' subset node in the subset in some enumeration thereof).   
\end{proof}

On the other hand, the maximization problems are similar in that maximizing the rank is generally consistent with maximizing the rank percentile: make as many large subsets as possible.

\begin{theorem} \label{max_hardness}
    {\sc Rank Maximization} and {\sc Rank Percentile Maximization} are {\sc NP}-hard.
\end{theorem}

\begin{proof}
    We reduce from {\sc Multiway Number Partitioning}. We are given an instance $\langle B,k \rangle$, where $B$ is a multiset of positive integers summing to $T$ and $k$ represents the number of multisets in our partition of $B$ (so $k|T$). We now reduce to an instance of {\sc Rank Maximization}. Let $V$ be equal to $B$ plus a special element $v^*$ (create multiple vertices for duplicate numbers), $E = V \times V$, and $S^* = \{v^*\}$. Define $\mu$ by $\mu(\{v^*\}) = \frac{T}{k}$ and otherwise $\mu(\{v\})$ is equal to the integer in $B$ corresponding to the element $v$. Now if the instance of {\sc Multiway Number Partitioning} has a solution, this establishes that the rank of $S^*$ may be at least $k+1$, and if such a rank is achieved, by positivity it must be the case that each of the $k$ subsets with greater value than $S^*$ under $\mu$ have value exactly equal to $\frac{T}{k}$. The reduction to {\sc Rank Percentile Maximization} is similar, except we must perturb $\mu$ slightly so that $S^*$ has value strictly smaller than the other subsets (this also establishes hardness of rank maximization when the inequality in the definition of rank is taken to be strict). For this, now let $\mu(S^*) = \frac{T}{k} - \frac{1}{2}$. If the instance of {\sc Multiway Number Partitioning} has a solution, then a percentile of $\frac{k+.5}{k+1}$ is possible. Now we extend to the other direction. By integrality of $B$, for $k$ subsets to be greater than $S^*$ under $\mu$, we need each to still have value $\frac{T}{k}$, and this is the maximum number of subsets possible, hence uniquely maximizing the percentile of $S^*$ (since by inspection, the percentile is a convex combination of 1, .5, and 0, with the former having the greatest coefficient when the largest number of subsets have value greater than $S^*$ and all subsets of $V \backslash S^*$ are of this form, making the coefficient of the latter 0). Thus the percentile of $\frac{k+.5}{k+1}$ is only achieved by a solution to the instance of {\sc Multiway Number Partitioning}.   
\end{proof}

Moreover, in a somewhat restricted set of examples, the maximization problems are hard to solve even approximately. 

\begin{theorem} \label{max_inapprox}
    If {\sc P} $\neq$ {\sc NP}, then for any $\epsilon > 0$ there is no polynomial-time $(\frac{3}{2}-\epsilon)$-approximation algorithm for {\sc Rank Maximization}.   
\end{theorem}

\begin{proof}
    Taking $k = 2$ in {\sc Multiway Number Partitioning}, we obtain the {\sc Partition} problem, which is still {\sc NP}-hard. Perform the same reduction to {\sc Rank Maximization} as in the proof for Theorem \ref{max_hardness}. Now if the {\sc Partition} instance has a solution, the optimal value for the {\sc Rank Maximization} instance is 3. Note that for any $\delta > 0$, $(\frac{2}{3} + \delta)\cdot3 > 2$, meaning that such an approximation algorithm must solve {\sc Partition} exactly (if the {\sc Partition} instance does not have a solution, the optimal value for the {\sc Rank Maximization} instance is less than 3, so an approximation algorithm would also report this). This would imply that {\sc P} $=$ {\sc NP}.  
\end{proof}

\begin{theorem} \label{max_percentile_inapprox}
    If {\sc P} $\neq$ {\sc NP}, then for any $\epsilon > 0$ there is no polynomial-time $(\frac{10}{9}-\epsilon)$-approximation algorithm for {\sc Rank Percentile Maximization}.   
\end{theorem}

\begin{proof}
    The proof is similar to that for \ref{max_inapprox}, except we apply the reduction to {\sc Rank Percentile Maximization} from Theorem \ref{max_inapprox} and need to compute the ratio of the optimal value in both cases. If the {\sc Partition} instance has a solution, the maximum percentile is $\frac{2.5}{3}$, whereas otherwise it is at most $\frac{1.5}{2}$ (this is the case where one large subset is possible, and we make that the only other subset in the partition; it is always optimal to eliminate a small subset compared to leaving it in the partition, and medium subsets are not possible since $\mu(S^*)$ is not integral). Lastly, we observe that $\frac{1.5/2}{2.5/3} = \frac{9}{10}$.      
\end{proof}

A subtlety arises in our analogue between the above maximization problems when medium subsets are present (which act as a hybrid between large subsets and small subsets in the percentile optimization problems), but we may still convert {\sc Rank Percentile Maximization} to {\sc Rank Maximization} without too much loss of optimality.

\begin{lemma} \label{max_approx}
    There is a 2-approximation algorithm for {\sc Rank Percentile Maximization} by reducing the problem to multiple instances of {\sc Rank Maximization} in polynomial-time. Moreover, the algorithm is optimal in the case that the maximum percentile is at least .5.
\end{lemma}

\begin{proof}
    The high-level idea is to take the better of the two percentiles obtained by (1) maximizing the number of large subsets in the partition; and (2) maximizing the number of large or medium subsets in the partition.
    Let $l^*$, $m^*$, and $s^*$ denote the number of large, medium, and small subsets in a valid partition achieving maximum percentile $p^*$, respectively. Since Lemma \ref{part_comb} implies that combining a small subset with any other subset increases the percentile under any partition, $s^*$ must be equal to the number of connected components of $G$ that have value (as a subset of vertices) less than $\mu(S^*)$ under $\mu$. Now applying Lemma \ref{frac_avg}, we can express $p^*$ as $1\cdot\frac{l^*}{l^*+m^*+s^*+1} + .5\cdot\frac{m^*+1}{l^*+m^*+s^*+1}$. \\ 
    
    First suppose that the first term is larger. Now create an instance of {\sc Rank Maximization} with input identical to the instance of {\sc Rank Percentile Maximization} except that the new measure $\tilde{\mu}$ is equal to $\mu$ on sets having empty intersection with $S^*$ and at $S^*$ is $\frac{a'}{b_1b_2\cdots b_n}$, where $\mu(\{v_1\}) = b_1, \mu(\{v_2\}) = b_2,\dots,\mu(\{v_n\}) = b_n$ and $a'$ is the smallest integer so that the expression is greater than $\mu(S^*)$ (which can be found by solving $\frac{x}{b_1b_2\cdots b_n} = \mu(S^*)$ and rounding up). Observe that maximizing rank here is equivalent to maximizing the number of large subsets in our original instance. Consider the partition returned by solving the {\sc Rank Maximization} instance, and further minimize the number of small and medium subsets (with respect to $\mu$) by absorbing them into large subsets within the given connected component of $G$ or, if there are no such subsets in the connected component, leaving it as one subset. Let $l_1$, $m_1$, and $s_1$ denote the number of large, medium, and small subsets in the resultant partition, respectively. By the preceding sentence, we have $l_1 \geq l^*$, $m_1 \leq m^*$, and $s_1 = s^*$, so the resultant partition has percentile $1\cdot\frac{l_1}{l_1+m_1+s_1+1} + .5\cdot\frac{m_1+1}{l_1+m_1+s_1+1} \geq 1\cdot\frac{l^*}{l^*+m^*+s^*+1} \geq \frac{1}{2}(1\cdot\frac{l^*}{l^*+m^*+s^*+1} + .5\cdot\frac{m^*+1}{l^*+m^*+s^*+1})$. \\

    If instead the second term is larger, we create an instance of {\sc Rank Maximization} whose input is completely identical to the instance of {\sc Rank Percentile Maximization}. Consider the partition returned by solving the {\sc Rank Maximization} instance, and further minimize the number of small subsets as above. Let $l_2$, $m_2$, and $s_2$ denote the number of large, medium, and small subsets in the resultant partition, respectively. Observe that $l_2 + m_2 \geq l^* + m^*$ and $s_2 = s^*$, so the resultant partition has percentile 
    \begin{align*}
    1\cdot\frac{l_2}{l_2+m_2+s_2+1} + .5\cdot\frac{m_2+1}{l_2+m_2+s_2+1} 
    &= .5\cdot\frac{l_2}{l_2+m_2+s_2+1} + .5\cdot\frac{l_2+m_2+1}{l_2+m_2+s_2+1} \\
    &\geq .5\cdot\frac{l^*+m^*+1}{l^*+m^*+s^*+1} \\
    &\geq .5\cdot\frac{m^*+1}{l^*+m^*+s^*+1} \\
    &\geq \frac{1}{2}\left(1\cdot\frac{l^*}{l^*+m^*+s^*+1} + .5\cdot\frac{m^*+1}{l^*+m^*+s^*+1}\right).
    \end{align*}
    

    Therefore, by returning the partition with the larger percentile of the two partitions constructed above, we guarantee that it has at least $\frac{1}{2}$ of the maximum percentile $p^*$. \\ 
    
    Lastly, if $p^* \geq .5$, then since $l \leq 1, m \geq 0, s \geq 0 \implies \frac{l+.5m-1}{l+m+s-1} \leq .5$, there exists a partition attaining percentile $p^*$ that does not have any medium or small subsets in a connected component that has value greater than $\mu(S^*)$ under $\mu$ (since for any partition with percentile at least .5, Lemma \ref{part_comb} implies combining $l$ large subsets, $m$ medium subsets, and $s$ small subsets in the above ranges may only increase the percentile). By inspection, the first partition constructed above is such a partition, so it will be returned by our algorithm (or another partition of the same percentile). Conversely, if the algorithm returns a partition with percentile at least .5, then it must be the case that $p^* \geq .5$, so our algorithm returns the maximum percentile by our preceding analysis.       
\end{proof}



\section{Algorithms for Special Cases} \label{special_cases}

\subsection{Complete Graph Case--Bounding Rank and Percentile}

The first special case of the problem we consider is where every pair of objects is similar, i.e. $G$ is a complete graph. The \textit{complete graph} case enables us to find (sharp) bounds on the rank and percentile and, correspondingly, may be used to model actors who optimize rank or percentile based on ex-post descriptions of similarity/categorization (instead of actually being constrained by it a priori). \\

We start with the minimization problems, both of which admit closed-form solutions in this setting. 

\begin{corollary}
    If $G$ is a complete graph, the minimum rank of $S^*$ under a valid partition is 1 if no vertex is large (as a singleton) and 2 otherwise. 
\end{corollary}

\begin{proof}
    The statement is a direct application of Corollary \ref{rank_lower_bound}. 
\end{proof}

The next statement tells us that, in the salient regime of this special case, the minimum percentile is on the order of $\frac{1}{s_0}$, but may become close to .5 if there is not both a subset that is larger and a subset that is smaller than $S^*$.    

\begin{theorem}
    Let $l_0$, $m_0$, and $s_0$ be the number of vertex singletons that are large, medium, and small, respectively. If $G$ is a complete graph, the minimum rank of $S^*$ under a valid partition is $\frac{\mathbf{1}_{\{l_0 > 0\}} + .5}{\mathbf{1}_{\{l_0 > 0\}} + \mathbf{1}_{\{l_0 = 0, m_0 > 0\}} + s_0 + 1} \leq .5$ if ($l_0 \geq 1$ and $s_0 \geq 1$) or ($m_0 \geq \frac{2s_0+1}{s_0-1}$, and $s_0 \geq 2$), $\frac{.5m_0 + .5}{m_0 + s_0 + 1} \leq .5$ if ($l_0 = 0$ or $s_0 \geq 1$) and ($m_0 < \frac{2s_0+1}{s_0-1}$ or $s_0 \leq 1$), and $\frac{.5m_0 + 1.5}{m_0 + 2} > .5$ if $l_0 \geq 1$ and $s_0 = 0$. 
\end{theorem}

\begin{proof}
    In each case, we exhibit a partition achieving the given percentile and argue that it is the minimum. The partition is formed as follows, starting from all singletons. Combine all large singletons into one subset. Then, if ($l_0 \geq 1$ and $s_0 \geq 1$) or ($l_0 = 0$, $m_0 \geq \frac{2s_0+1}{s_0-1}$, and $s_0 \geq 2$), add all medium singletons into the (possibly empty) subset comprised of all large singletons. \\

    By Lemma \ref{part_comb}, there must not be more than one large subset in the final partition (since combining two of them necessarily decreases the percentile, and we may always do this in a complete graph), so all large vertices must be in the same subset in the partition. Similarly, all small singletons must be subsets in the final partition. Therefore, the only degrees of freedom we have are combining medium subsets or combining medium subset(s) and the large subset (if it exists). Consider the updated partition we obtain after the first step (aggregating all large singletons). If $l_0 \geq 1$ and $s_0 = 0$, then the updated partition achieves percentile greater than .5, and applying Lemma \ref{part_comb} with $l \leq 1$ and $s = 0$ reveals that no coarser partition has smaller percentile. If ($l_0 = 0$ or $s_0 \geq 1$) and ($m_0 < \frac{2s_0+1}{s_0-1}$ or $s_0 \leq 1$), then by the first condition the updated percentile is at most .5, and the second is equivalent to $\frac{.5m_0-1}{m_0-1} < \frac{.5m_0+.5}{m_0+s_0+1}$, meaning that combining the medium subsets would increase the percentile per Lemma \ref{part_comb} (or any subset thereof, since the left-hand side of the inequality is increasing in $m_0$, and note that there are no large subsets to combine with in this case). Lastly, if ($l_0 \geq 1$ and $s_0 \geq 1$) or ($m_0 \geq \frac{2s_0+1}{s_0-1}$, and $s_0 \geq 2$), then by the first condition the updated percentile is at most .5, and the second is equivalent to $\frac{.5m_0-1}{m_0-1} \geq \frac{.5m_0+.5}{m_0+s_0+1}$, so by Lemma \ref{part_comb}, combining the medium singletons is optimal even if there are no large subsets, and if there is a large subset, it is optimal to add the medium singletons to that since doing so eliminates the maximum number of medium subsets, which is desirable since the updated percentile is at most .5, while not adding any new large subsets, which would increase the percentile. Thus in all cases, the algorithm computes an optimal partition, and by inspection this yields the claimed percentile. 
\end{proof} 

    


On the other hand, the maximization problems still have inherent intractability (the inapproximability results, Theorem \ref{max_inapprox} and Theorem \ref{max_percentile_inapprox}, were found in this special case), but may now be approximated within a constant factor. \\

The following definitions will be useful for our approach.

\begin{definition} \label{intermediate_tiny_def}
    An element $v \in V$ \textit{intermediate} if $\frac{1}{2}\mu(S^*) \leq \mu(\{v\}) < \mu(S^*)$ and \textit{tiny} if $\mu(\{v\}) < \frac{1}{2}\mu(S^*)$. Tiny elements have two special cases of interest. Given a constant $\frac{1}{2}\mu(S^*) < c < \mu(S^*)$, we say an element $v \in V$ is \textit{subintermediate} if $\mu(S^*)-c \leq \mu(\{v\}) < \frac{1}{2}\mu(S^*)$ and \textit{subsubintermediate} if $\frac{1}{2}(\mu(S^*)-c) \leq \mu(\{v\}) < \mu(S^*)-c$.
\end{definition}



\begin{theorem} \label{complete_max_approx}
    The complete graph case of {\sc Rank Maximization} has a polynomial-time approximation algorithm with approximation guarantee $\mathrm{OPT} \leq \frac{3}{2}\mathrm{ALG} + 1 - \frac{1}{2}(l_0+m_0) - \frac{1}{2}i^*$, where $l_0$ and $m_0$ are the number of vertex singletons that are large and medium, respectively, and $i^*$ is the number of subsets comprised of intermediate elements in an arbitrary optimal solution.  
\end{theorem} 

\begin{proof}
    Two key insights for our algorithm is that (1) a subset needs exactly 2 intermediate elements to become medium or large if no tiny elements are used; and (2) medium or large subsets comprised entirely of tiny elements never overshoot the threshold $\mu(S^*)$ by a factor exceeding $\frac{3}{2}$. We will also need to analyze subsets combining intermediate and tiny elements, since the tiny elements may exceed the \textit{remainder} of the threshold $\mu(S^*)$ by a factor exceeding $\frac{3}{2}$, which affects how many subsets may be formed from remaining tiny elements. These dynamics motivate Definition \ref{intermediate_tiny_def}. \\

    We now describe the algorithm and then state it formally as Algorithm \ref{complete_rank_max_alg}. Suppose there are exactly $r$ intermediate elements $\{v_1,\dots,v_r\} \in V$ and index $V$ such that $\mu(\{v_1\}) \geq \cdots \geq \mu(\{v_n\})$. Initialize $P = \emptyset$. For each $v \in V$ that is medium or large (as a singleton set), add $\{v\}$ to $P$. For each $i \in [\lceil \frac{2r}{3} \rceil, r]$ that is even, add $\frac{i}{2}$ subsets to $P$, each containing 2 intermediate elements (and nothing else) from the set $\{v_{r-i+1},\dots,v_r\}$. Now initialize a subset $S_j = \{v_j\}$ for each of the $r-i \in [0, \lfloor \frac{r}{3} \rfloor]$ remaining intermediate elements $v_j: 1 \leq j \leq r-i$ and add it to $P$. Set $c = \mu(\{v_{r-i}\})$ and let $s_1$ and $s_2$ be the number of subintermediate and subsubintermediate elements of $V$, respectively. For each tuple of natural numbers $(x,y,z)$ such that $0 \leq x \leq s_1$, $0 \leq 2y + z \leq s_2$, and $x+y+z \leq r-i$, add 1 distinct subsubintermediate element to $S_1,\dots,S_z$, add 2 distinct subsubintermediate elements to $S_{z+1},\dots,S_{y+z}$, and add 1 distinct subintermediate element to $S_{z+y+1},\dots,S_{x+y+z}$ (all in ascending value order). For $1 \leq l \leq z$, add any tiny elements not added to any subset thus far one-by-one to $S_l$ until $\mu(S_l) \geq \mu(S^*)$ (such a subset is called \textit{completed}). Lastly, add any tiny elements not added to any subset thus far one-by-one to an empty subset until its value is at least $\mu(S^*)$, transitioning to a new empty subset whenever this criterion is met. Add each subset that meets this criterion to $P$; add a subset containing any remaining elements of $V$ to $P$. The algorithm returns $\mathrm{argmax}_{i,x,y,z} P$. \\

\begin{algorithm}[H]
\caption{Given an instance $(V = \{v_1,\dots,v_n,\mu,S^*)$ of {\sc Rank Maximization} in the complete graph case such that $\mu(\{v_1\}) \geq \cdots \geq \mu(\{v_n\})$, compute the maximum rank of $S^*$}\label{complete_rank_max_alg}
\begin{algorithmic}

\State Initialize $P(\cdot,\cdot,\cdot,\cdot) \gets \emptyset$ \Comment{initialize empty partitions}
\State Add each large or medium $v \in V$ to $P(\cdot,\cdot,\cdot,\cdot)$ as a singleton
\For{$i$ in $\left[\lceil \frac{2r}{3} \rceil, r\right]: i$ even}
\State Split $\{v_{r-i+1},\dots,v_r\}$ into disjoint pairs and add each pair to $P(i,\cdot,\cdot,\cdot)$
\State $T_1(i,\cdot,\cdot,\cdot) \gets \{v \in V: v \text{ is subintermediate relative to } \mu(\{v_{r-i}\})\}$ 
\State $T_2(i,\cdot,\cdot,\cdot) \gets \{v \in V: v \text{ is subsubintermediate relative to } \mu(\{v_{r-i}\})\}$
\For{$j$ in $[1,r-i]$} 
\State $S_j(i,\cdot,\cdot,\cdot) \gets \{v_j\}$ 
\EndFor
\For{$(x, y, z) \in \mathbb{N}^3$ : $x \leq |T_1(\cdot,\cdot,\cdot,\cdot)| \wedge 2y + z \leq |T_2(\cdot,\cdot,\cdot,\cdot)| \wedge x + y + z \leq r - i$}
\For{$j$ in $[1,z]$}
\State Move $e \in T_2(i,x,y,z)$ with smallest $\mu$-measure to $S_j(i,x,y,z)$
\EndFor
\For{$j$ in $[z+1,y+z]$}
\State Move $e_1,e_2 \in T_2(i,x,y,z)$ distinct with smallest $\mu$-measure to $S_j(i,x,y,z)$
\State $P(i,x,y,z) \gets P(i,x,y,z) \cup \{S_j(i,x,y,z)\}$
\EndFor
\For{$j$ in $[y+z+1,x+y+z]$}
\State Move $e \in T_1(i,x,y,z)$ with smallest $\mu$-measure to $S_j(i,x,y,z)$
\State $P(i,x,y,z) \gets P(i,x,y,z) \cup \{S_j(i,x,y,z)\}$
\EndFor
\State $R(i,x,y,z) \gets T_1(i,x,y,z) \cup T_2(i,x,y,z) \cup \{v_{r+|T_1(\cdot,\cdot,\cdot,\cdot)|+|T_2(\cdot,\cdot,\cdot,\cdot)|},\dots,v_n\}$
\For{$j$ in $[1,z]$}
\While{$\mu(S_j(i,x,y,z)) < \mu(S^*)$ and}
\State Move an element $e \in R(i,x,y,z)$ to $S_j$
\EndWhile
\EndFor
\While{$R(i,x,y,z) \neq \emptyset$}
\State $S \gets \emptyset$
\While{$R(i,x,y,z) \neq \emptyset$ and $\mu(S) < \mu(S^*)$}
\State Move an element $e \in R(i,x,y,z)$ to $S$
\EndWhile
\State $P \gets P \cup \{S\}$
\EndWhile
\EndFor
\EndFor
\State \Return $\mathrm{argmax}_{i,x,y,z} P(i,x,y,z)$
\end{algorithmic}
\end{algorithm}

    To prove the approximation claim, consider an optimal partition $P^*$ achieving a rank of $\mathrm{OPT}$ and write $\mathrm{OPT} = A^* + B_0^* + B^* + C^* + 1$, where $A^*$, $B_0^* + B^*$, and $C^*$ are the number of subsets in $P^*$ containing whose largest element is large or medium, intermediate, and tiny, respectively, and where $B_0^*$ is the number of subsets in $P^*$ containing only intermediate elements. Similarly, we write $\mathrm{ALG} = A + B_0 + B + C + 1$ for the rank achieved by the partition $\mathrm{argmax}_{i,x,y,z} P$ returned by our algorithm. We now perform a term-by-term comparison of $\mathrm{OPT}$ and $\mathrm{ALG}$. \\

    First, we argue by exchange that $P^*$ may be chosen such that each medium or large element forms a singleton subset; if $P^*$ is not in this form, we may iteratively form a singleton subset from a medium or large element that is in a subset in $P^*$ containing other elements. Since any subset $S$ with a medium or large element satisfies $\mu(S) \geq \mu(S^*)$, this does not decrease the rank achieved by $P^*$. Thus by construction, $A^* = A = l_0 + m_0$. \\

    Next, write $S^*_j$, $1 \leq j \leq B^*$ for the subsets of $B^*$ with an intermediate element $v^*_j$; reorder such that $\mu(\{v^*_1\}) \geq \cdots \geq \mu(\{v^*_{B^*}\})$. Note that $B^* \leq r$. Consider the iteration of our algorithm where $r-i \in \{\lfloor \frac{B^*}{3} \rfloor,\lfloor \frac{B^*}{3} \rfloor+1\}$, $x = x^*$, $y = y^*$, and $z = z^*$, where $x^*$, $y^*$, and $z^*$ are the number of subsets in $\{S^*_1,\dots,S^*_{B^*}\}$ containing exactly 1 subintermediate element, exactly 1 subsubintermediate element, and exactly 2 subsubintermediate elements, respectively (still defined relative to $\mu(\{v_{r-i}\})$). Suppose for now that $r-i = \lfloor \frac{B^*}{3} \rfloor$. After termination of the algorithm, define $$W := \sum_{j=1}^{\lfloor \frac{B^*}{3} \rfloor} \mu(S_j \backslash \{v_j\})$$ and $$W^* := \sum_{j=1}^{B^*} \mu(S^*_j \backslash \{v^*_j\})$$ for the value of tiny elements added to these intermediate singletons by our algorithm and an optimal algorithm, respectively. We want to show $W \leq \frac{3}{4}W^*$, which will imply that all of $S_1,\dots,S_{\lfloor \frac{B^*}{3} \rfloor}$ were able to be completed by our algorithm (otherwise, we would have added more elements to try to complete a remaining subset). \\

    At a high level, this iteration of our algorithm splits $B^*$ into thirds, choosing to pair the smallest two-thirds of intermediate elements with each other while still completing a subset for each of the remaining largest one-third of intermediate elements (using tiny elements). This ensures that our algorithm still captures two-thirds of the subsets from $B^*$, and classifying tiny elements relative to the value that the smallest intermediate element in our largest one-third needs to be completed allows us to bound the value of tiny elements used with intermediate elements relative to an optimal algorithm, resulting in similar tiny element values ``left over'' to be combined for more subsets. The resulting partition, and how it differs from the optimal partition, is depicted in Figure \ref{fig:3_2_approximation_visual}.

    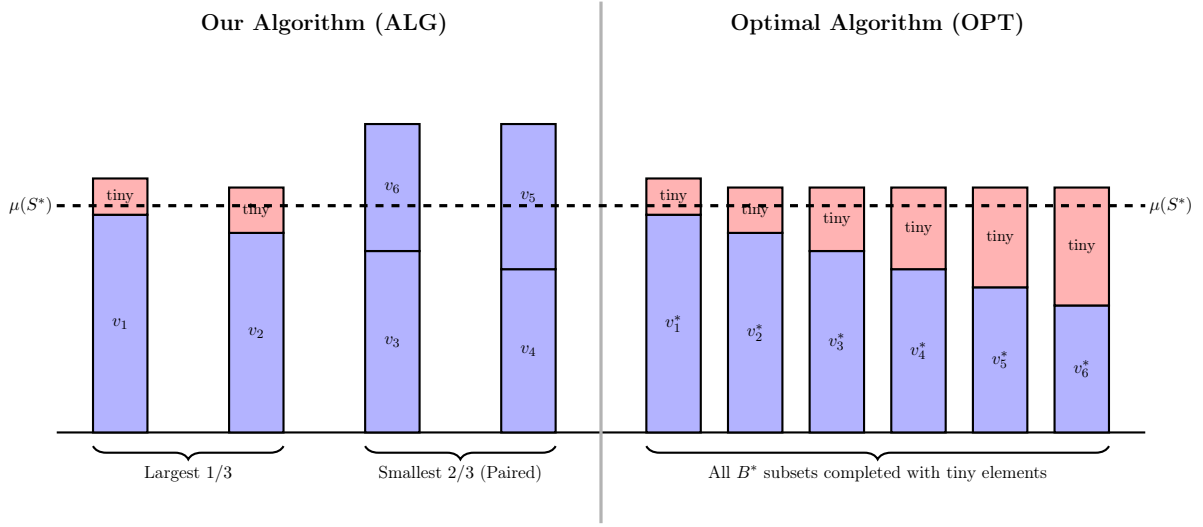
\begin{figure}[H]
    \centering
    \begin{tikzpicture}[scale=0.6, transform shape]
        \def\W{1.2}         
        \def\gapL{1.8}      
        \def\gapR{0.6}      
        
        \def\hA{4.8}
        \def\hB{4.4}
        \def\hC{4.0}
        \def\hD{3.6}
        \def\hE{3.2}
        \def\hF{2.8}
        
        \def\tA{0.8}
        \def\tB{1.0}
        \def\tC{1.4}
        \def\tD{1.8}
        \def\tE{2.2}
        \def\tF{2.6}

        \draw[thick] (-12, 0) -- (12, 0); 

        \node[font=\Large\bfseries] at (-6.1, 9) {Our Algorithm (ALG)};
        
        \def\xLone{-11.2}
        \def\xLtwo{-8.2}   
        \def\xLthree{-5.2} 
        \def\xLfour{-2.2}  

        \filldraw[fill=blue!30, draw=black, thick] (\xLone,0) rectangle ++(\W,\hA) node[midway] {$v_1$};
        \filldraw[fill=red!30, draw=black, thick] (\xLone,\hA) rectangle ++(\W,\tA) node[midway] {\small tiny};

        \filldraw[fill=blue!30, draw=black, thick] (\xLtwo,0) rectangle ++(\W,\hB) node[midway] {$v_2$};
        \filldraw[fill=red!30, draw=black, thick] (\xLtwo,\hB) rectangle ++(\W,\tB) node[midway] {\small tiny};

        \filldraw[fill=blue!30, draw=black, thick] (\xLthree,0) rectangle ++(\W,\hC) node[midway] {$v_3$};
        \filldraw[fill=blue!30, draw=black, thick] (\xLthree,\hC) rectangle ++(\W,\hF) node[midway] {$v_6$};

        \filldraw[fill=blue!30, draw=black, thick] (\xLfour,0) rectangle ++(\W,\hD) node[midway] {$v_4$};
        \filldraw[fill=blue!30, draw=black, thick] (\xLfour,\hD) rectangle ++(\W,\hE) node[midway] {$v_5$};

        \draw[decorate, decoration={brace, amplitude=5pt, mirror}, thick] (\xLone, -0.3) -- (\xLtwo+\W, -0.3) node[midway, below=8pt] {Largest $1/3$};
        \draw[decorate, decoration={brace, amplitude=5pt, mirror}, thick] (\xLthree, -0.3) -- (\xLfour+\W, -0.3) node[midway, below=8pt] {Smallest $2/3$ (Paired)};

        \draw[very thick, gray!60] (0, -2) -- (0, 9.5);

        \node[font=\Large\bfseries] at (6.1, 9) {Optimal Algorithm (OPT)};
        
        \def\xRone{1.0}
        \def\xRtwo{2.8}   
        \def\xRthree{4.6} 
        \def\xRfour{6.4}  
        \def\xRfive{8.2}  
        \def\xRsix{10.0}  

        \filldraw[fill=blue!30, draw=black, thick] (\xRone,0) rectangle ++(\W,\hA) node[midway] {$v^*_1$};
        \filldraw[fill=red!30, draw=black, thick] (\xRone,\hA) rectangle ++(\W,\tA) node[midway] {\small tiny};

        \filldraw[fill=blue!30, draw=black, thick] (\xRtwo,0) rectangle ++(\W,\hB) node[midway] {$v^*_2$};
        \filldraw[fill=red!30, draw=black, thick] (\xRtwo,\hB) rectangle ++(\W,\tB) node[midway] {\small tiny};

        \filldraw[fill=blue!30, draw=black, thick] (\xRthree,0) rectangle ++(\W,\hC) node[midway] {$v^*_3$};
        \filldraw[fill=red!30, draw=black, thick] (\xRthree,\hC) rectangle ++(\W,\tC) node[midway] {\small tiny};

        \filldraw[fill=blue!30, draw=black, thick] (\xRfour,0) rectangle ++(\W,\hD) node[midway] {$v^*_4$};
        \filldraw[fill=red!30, draw=black, thick] (\xRfour,\hD) rectangle ++(\W,\tD) node[midway] {\small tiny};

        \filldraw[fill=blue!30, draw=black, thick] (\xRfive,0) rectangle ++(\W,\hE) node[midway] {$v^*_5$};
        \filldraw[fill=red!30, draw=black, thick] (\xRfive,\hE) rectangle ++(\W,\tE) node[midway] {\small tiny};

        \filldraw[fill=blue!30, draw=black, thick] (\xRsix,0) rectangle ++(\W,\hF) node[midway] {$v^*_6$};
        \filldraw[fill=red!30, draw=black, thick] (\xRsix,\hF) rectangle ++(\W,\tF) node[midway] {\small tiny};
        
        \draw[decorate, decoration={brace, amplitude=5pt, mirror}, thick] (\xRone, -0.3) -- (\xRsix+\W, -0.3) node[midway, below=8pt] {All $B^*$ subsets completed with tiny elements};

        
        \draw[dashed, very thick, black] (-12, 5.0) -- (12, 5.0)
            node[pos=0, left, fill=white, inner sep=2pt, text=black] {$\mu(S^*)$}
            node[pos=1, right, fill=white, inner sep=2pt, text=black] {$\mu(S^*)$};

    \end{tikzpicture}
    \caption{A visual representation for one component of the $\frac{3}{2}$-approximation guarantee involving intermediate elements. Our algorithm completes the largest one-third of the intermediate elements with tiny elements and pairs the remaining two-thirds with each other (among the intermediate elements not paired by the optimal algorithm), guaranteeing they cross the $\mu(S^*)$ threshold without relying on further tiny elements.}
    \label{fig:3_2_approximation_visual}
\end{figure}
    
    More formally, to compare $W$ and $W^*$, we assume without loss of generality that $\mu(S^*_1 \backslash \{v^*_1\}) \leq \cdots \leq \mu(S^*_{B^*} \backslash \{v^*_{B^*}\})$ (for any reversed pair of subsets, we may swap their non-intermediate elements while keeping each complete). In particular, $$\sum_{j=\lfloor \frac{B^*}{3} \rfloor+1}^{2\lfloor \frac{B^*}{3} \rfloor} \mu(S^*_j \backslash \{v^*_j\}) \leq \frac{1}{2}W^*,$$ so it suffices to show that $W \leq \frac{3}{2} \sum_{j=\lfloor \frac{B^*}{3} \rfloor+1}^{2\lfloor \frac{B^*}{3} \rfloor} \mu(S^*_j \backslash \{v^*_j\})$. Whenever our algorithm adds 1 subintermediate or 2 subsubintermediate elements to a subset, the subset becomes completed; because $x = x^*$, $y = y^*$, and we choose the smallest (with respect to value under $\mu$) such elements first, we have that each of $\mu(S_{z+1} \backslash \{v_{z+1}\}),\dots,\mu(S_{x+y+z} \backslash \{v_{x+y+z}\})$ is less than or equal to a corresponding, distinct summand in $\sum_{j=\lfloor \frac{B^*}{3} \rfloor+1}^{2\lfloor \frac{B^*}{3} \rfloor} \mu(S^*_j \backslash \{v^*_j\})$. As for the remaining summands in $W$, each corresponding subset must have been completed by an element with value less than $\frac{1}{2}(\mu(S^*)-c) \equiv \frac{1}{2}(\mu(S^*)-\mu(v_{\lfloor \frac{B^*}{3} \rfloor}))$, whereas the subset is necessarily completed after adding a total value of $(\mu(S^*)-\mu(v_{\lfloor \frac{B^*}{3} \rfloor}))$, so those summands are at most $\frac{3}{2}(\mu(S^*)-\mu(v_{\lfloor \frac{B^*}{3} \rfloor}))$. On the other hand, each summand in $\sum_{j=\lfloor \frac{B^*}{3} \rfloor+1}^{2\lfloor \frac{B^*}{3} \rfloor} \mu(S^*_j \backslash \{v^*_j\})$ is at least $\mu(S^*)-\mu(v_{\lfloor \frac{B^*}{3} \rfloor})$. Comparing summand-by-summand, we conclude $$W \leq \frac{3}{2} \sum_{j=\lfloor \frac{B^*}{3} \rfloor+1}^{2\lfloor \frac{B^*}{3} \rfloor} \mu(S^*_j \backslash \{v^*_j\}).$$ Since the algorithm completes $S_1,\dots,S_{\lfloor \frac{B^*}{3} \rfloor}$ and pairs the other $\lceil \frac{2B^*}{3} \rceil$ intermediate elements to form $\lfloor \frac{\lceil \frac{2B^*}{3} \rceil}{2} \rfloor$ subsets, $B_0^* + B^* \leq B_0 + \frac{3}{2}B + 1$. \\ 
    
    Lastly, after the algorithm completes the claimed subsets, the elements not placed in subsets have total value $\mu(V) - \sum_{i=1}^r \mu(\{v_r\}) - W$, which is at least $\mu(V) - \sum_{i=1}^r \mu(\{v_r\}) - W^*$, the maximum total value of elements in the subsets of $P^*$ that contain only tiny elements. Our algorithm deals with this value by adding elements to subsets until they become completed, which occurs once their value reaches $\mu(S^*)$. Since each tiny element has value less than $\frac{1}{2}\mu(S^*)$, each completed subset has value less than $\frac{3}{2}\mu(S^*)$, whereas each medium or large subset in $P^*$ containing only tiny elements has value at least $\mu(S^*)$. Thus $C^* < \frac{3}{2}C + \frac{1}{2}$. \\

    If instead $r-i = \lfloor \frac{B^*}{3} \rfloor+1$, the algorithm attempts to complete an additional subset that already has an intermediate element before creating a new subset with tiny elements, so by a greedy stays ahead argument we still get $B_0^* + B^* + C^* < B_0 + \frac{3}{2}B + 1 + \frac{3}{2}C + \frac{1}{2}$. Combining the (in)equalities derived for $A,A^*,B_0,B_0^*,B,B^*,C,C^*$ yields $\mathrm{OPT} < \frac{3}{2}\mathrm{ALG} + \frac{3}{2} - \frac{1}{2}(l_0+m_0) - \frac{1}{2}B_0^*$, which is equivalent to the desired claim since the right-hand side is a multiple of $\frac{1}{2}$.
\end{proof}

\begin{corollary}
    The complete graph case of {\sc Rank Percentile Maximization} has a polynomial-time $\max\left(1,\frac{\frac{3}{2}l' - \frac{1}{2}(l_0+m_0) - \frac{1}{2}i^* + \frac{3}{2}}{\frac{3}{2}l' - \frac{1}{2}(l_0+m_0) - \frac{1}{2}i^* + \frac{5}{2}}\left(1+\frac{1}{l'}\right)\right)$-approximation algorithm, where $l_0$ and $m_0$ are the number of vertex singletons that are large and medium, respectively, $i^*$ is the number of subsets comprised of intermediate elements in an arbitrary optimal solution, and $l'$ is the number of large subsets, and $l'$ is the number of large subsets returned by the algorithm's rank maximization subroutine. 
\end{corollary}

\begin{proof}
    We consider 3 cases. If $\mu(V \backslash S^*) < \mu(S^*)$, any partition of $V$ is comprised of small subsets, the percentile of which is maximized by having just 1 small subset, so the partition $(V \backslash S^*,S^*)$ maximizes the percentile of $S^*$. Similarly, if $\mu(V \backslash S^*) = \mu(S^*)$, any valid partition of $V$ is comprised of small subsets or one medium subset, the percentile of which is maximized by having just 1 medium subset, so the partition $(V \backslash S^*,S^*)$ maximizes the percentile of $S^*$. If $\mu(V \backslash S^*) > \mu(S^*)$, there exists a partition with a large subset, and since all partitions are valid, the maximum percentile is achieved by a partition containing $l^*$ large subsets, 0 medium subsets, and 0 small subsets, where $l^*$ is the maximum number of mutually disjoint subsets possible (by the proof of Lemma \ref{max_approx}). In other words, the maximum percentile is $\frac{l^*}{l^*+1}$. Instead, we use the algorithm described for Theorem \ref{complete_max_approx} to obtain $l'$ large subsets (after applying the preprocessing step in the proof of Lemma \ref{max_approx} to preclude medium subsets) and combine any remaining elements into one of these subsets, obtaining a percentile of $\frac{l'+1}{l'}$. Hence our approximation ratio in this case is $\frac{l^*}{l^*+1}\cdot\frac{l'+1}{l'}$. We can upper bound the first fraction by adding the same thing to the numerator and denominator, replacing $l^*$ with its upper bound given by the right-hand side of the approximation guarantee in Theorem \ref{complete_max_approx} plus $\frac{1}{2}$ (since we no longer add 1 to get the rank). This gives the desired factor.       
\end{proof}

\subsection{Linear Component Case}

Sometimes, the objects we wish to rank may be organized within each connected component by a single variable, with the similar elements in a connected component being those that have adjacent or consecutive values of the variable. For instance, each connected component may correspond to a television series, with the elements being its episodes chronologically ordered. Subsets in $\mathcal{S}$ may then be interpreted as seasons of a given series, and the above computational problems correspond to optimizing the rank of $S^*$, (a season of) a fixed series, with respect to hours watched, number of award nominations, etc. While $\mathcal{S}$ still contains exponentially many subsets in general, all four of the above problems are polynomial-time solvable in this special case where all connected components of $G$ are line graphs, called the \textit{linear component} case. \\

In this case, we will write $V = \{v_1,\dots,v_n\}$ to mean that each vertex $v_i$ is adjacent to at most $v_{i-1}$ and $v_{i+1}$ (if they exist) and no other vertices, so each connected component of $G$ is a line graph of the form $\{v_j,\dots,v_{j+k-1}\}$, where $k$ is the number of vertices in the component. Thus $S^*$ constitutes a line subgraph in one of the connected components, though as a preprocessing step for the rank percentile algorithms we discard $S^*$ from the description of $V$ and shift the index of proceeding vertices by $-|S^*|$. This truncates a linear component of $G$ or decomposes it into two linear components, so $G$ is still of the form argued above. \\

First, the linear component case of {\sc Rank Minimization} is just a special case of {\sc Rank Minimization} as defined above, so it is still polynomial-time solvable (see Theorem \ref{min_in_P}). 

\begin{corollary}
    The linear component case of {\sc Rank Minimization} is polynomial-time solvable. 
\end{corollary}

As for {\sc Rank Maximization}, this special case reduces to the classical unweighted interval scheduling problem. 

\begin{theorem}
    The linear component case of {\sc Rank Maximization} is polynomial-time solvable. 
\end{theorem}

\begin{proof}
    First suppose $G$ has exactly one connected component. Therefore, $G$ is a line graph with first vertex $v_1$ and last vertex $v_n$, and without loss of generality $S^* = \{v_1,\dots,v_k\}$. We reduce to unweighted interval scheduling as follows. For all $k < i \leq j \leq n$, if $\{v_i,\dots,v_j\}$ is medium or large, add a job occupying the interval $[i,j+.5]$. Thus the intervals are in one-to-one correspondence with the set of medium or large subsets of $V$, and two intervals being non-overlapping means that the corresponding subsets of $V$ are disjoint. As a result, maximizing the number of non-overlapping subsets is equivalent to maximizing the number of medium or large subsets with pairwise empty intersection, which is the same as maximizing the number of medium or large subsets in a partition of $V$ that includes $S^*$. If instead $G$ has multiple connected components, or $G \backslash S^*$ does as a result of $S^*$ being in the middle of a connected component, perform this reduction for each connected component and take the union of the optimal set of jobs over each instance; this still maximizes the number of medium or large subsets with pairwise empty intersection because it does so for each connected component, and a subset in $\mathcal{S}$ may not intersect multiple components. The preprocessing step is $O(n^2)$, and the greedy subroutine for unweighted interval scheduling is then $O(n^2 \log n)$, so our algorithm is $O(n^2 \log n)$.          
\end{proof}

For the percentile problems of this special case, we must also take care to the fact that when modifying a partition of $V$, the percentile of $S^*$ experiences ``inertia" based on how many subsets are already in the partition. However, due to the one-dimensional structure of subsets in this case, we may optimize the percentile with a relatively efficient dynamic programming approach.    

\begin{theorem} 
    The linear component case of {\sc Rank Percentile Minimization} and {\sc Rank Percentile Maximization} are polynomial-time solvable.
\end{theorem}

\begin{proof}
Consider the following algorithm for {\sc Rank Percentile Minimization} under this case. \\ 

\begin{algorithm}[H]
\caption{Given an instance $(V,E,\mu,S^*)$ of {\sc Rank Percentile Minimization} in the linear component case, compute the minimum percentile of $S^*$}\label{linear_percentile_alg}
\begin{algorithmic}

\For{$1 \leq i \leq n$}
    \For{$i < j \leq n$}
       \State $C[i,j] = 
        \begin{cases}
            \frac{1.5}{2} & \text{if } \{v_i,\dots,v_j\} \text{ is large} \\
            \frac{1}{2} & \text{if } \{v_i,\dots,v_j\} \text{ is medium} \\
            \frac{.5}{2} & \text{if } \{v_i,\dots,v_j\} \text{ is small} \\
            \infty & \text{otherwise} 
       \end{cases}$
    \EndFor
    \State $\mathrm{OPT}[i,1] = C[1,i]$
\EndFor

\For{$2 \leq k \leq n$} 
    \For{$1 \leq j \leq n$}
        \State $\mathrm{OPT}[j,k] = \frac{1}{k+1}(\inf_{1 \leq i < j} \{k\mathrm{OPT}[i,k-1] + 2C[i+1,j] - .5\})$
    \EndFor
\EndFor

\State \Return $\min_{1 \leq k \leq n} \{\mathrm{OPT}[n,k]\}$
        
\end{algorithmic}
\end{algorithm}

To see that this algorithm is correct, we argue that $\mathrm{OPT}[j,k]$ gives the minimum percentile of $S^*$ under a valid partition with exactly $k$ subsets other than $S^*$ in the problem instance where the vertex set is $\{v_1,\dots,v_j\} \cup S^*$ and the similarity relation is $E|_{(\{v_1,\dots,v_j\} \cup S^*) \times (\{v_1,\dots,v_j\} \cup S^*)}$. (We also adopt the convention that if there is no such partition, the minimum percentile is $\infty$.) First, when $c = 1$, notice that we have $\mathrm{OPT}[j,1] = C[1,j]$, which by definition gives the percentile of $S^*$ under the only partition of $\{v_1,\dots,v_j\} \cup S^*$ with 1 subset other than $S^*$: the trivial partition $(\{v_1,\dots,v_j\},S^*)$. (If $\{v_1,\dots,v_j\} \not\in \mathcal{S}$, then there is no such valid partition of $\{v_1,\dots,v_j\} \cup S^*$, and we set $C[1,j] = \infty$.) Suppose the claim holds for $k = k'-1 \in \mathbb{N}$ and all $1 \leq i < j$. Observe that if $k' \geq j$, then there any expression of the form $\mathrm{OPT}[i,k'-1]$ will be $\infty$ by hypothesis since we may not partition $i < j$ elements into $k'$ non-empty subsets; correspondingly, the infimum used to define $\mathrm{OPT}[j,k']$ is infinite. Otherwise, $k'-1 \leq j$. \\

Because $G$ has linear components, the valid partitions of $\{v_1,\dots,v_j\} \cup S^*$ with $k'$ subsets other than $S^*$ are exactly the valid partitions of $\{v_1,\dots,v_i\} \cup S^*$ plus the subset $[i+1,j]$, over all $1 \leq i < j$ such that $\{v_{i+1},\dots,v_j\} \in \mathcal{S}$. We claim that under the latter, $S^*$ has minimum percentile equal to $\frac{1}{k'+1}(k'\mathrm{OPT}[i,k'-1] + 2C[i+1,j] - .5)$. Indeed, $k'\mathrm{OPT}[i,k'-1]$ gives the number of large subsets plus .5 times the number of medium subsets in a given valid partition of $\{v_1,\dots,v_i\} \cup S^*$ plus .5; $2C[i+1,j]$, $[i+1,j]$. Now subtracting .5 and then multiplying by $\frac{1}{k'+1}$, we obtain the expression for the percentile of $S^*$ under a valid partition of $\{v_1,\dots,v_j\} \cup S^*$ with $k'$ subsets. Similarly, for given $i$ as above, the percentile under a partition with $k'$ subsets other than $S^*$ and the last subset $\{v_{i+1},\dots,v_j\}$ is a decreasing function of the percentile under the corresponding partition of $\{v_1,\dots,v_i\} \cup S^*$ with $k'-1$ subsets other than $S^*$, so the former is minimized by such a corresponding partition achieving minimum percentile, i.e., $\mathrm{OPT}[i,k'-1]$. Taking the infimum over all $1 \leq i < j$ results in minimizing the percentile under a partition of $\{v_1,\dots,v_j\} \cup S^*$ with $k'$ subsets other than $S^*$. (If the infimum is $\infty$, there is no valid partition of $\{v_1,\dots,v_j\} \cup S^*$ with $k'$ subsets other than $S^*$). \\

Finally, we note that there may be at most $n$ subsets in a partition of $V = \{v_1,\dots,v_n\}$, so $\min_{1 \leq k \leq n} \{\mathrm{OPT}[n,k]\}$ is the minimum percentile of $S^*$ under a valid partition of $V$. \\ 

To solve {\sc Rank Percentile Maximization}, we simply replace $\infty$ with $-\infty$, $\inf$ with $\sup$, and $\min$ with $\max$ in Algorithm \ref{linear_percentile_alg}; the proof of correctness is analogous. 
\end{proof}

\subsection{Uniform Value Case}

Another natural special case of the problem is where $\mu$ assigns the same value to each element of $V$; concretely, for all $v \in V$, $\mu(\{v\}) = 1$ (without loss of generality). Let $|S^*| = k \in \mathbb{N}$. Then, a subset $S$ is large if $|S| > |S^*|$, medium if $|S| = |S^*|$, and small if $|S| < |S^*|$. We call this the \textit{uniform value} case. \\

To see an application of this, think of each element of $V$ as the physical location of a specific type of resource, like a hospital or a polling place, and $E$ the set of roads connecting them. Now each $S \in \mathcal{S}$ can be thought of as a region bounded by roads, and the rank or percentile of $S^*$ under a partition indicates how well-resourced $S^*$ is as a region. \\

Due to Theorem \ref{min_in_P}, the uniform value case of {\sc Rank Minimization} is polynomial-time solvable, and because each element of $V$ has the same value, all singletons are small subsets, so the partition that minimizes rank also maximizes the number of subsets in a valid partition of $V$ while having no medium or large subsets, hence achieving the minimum percentile of $S^*$ as well. Therefore, {\sc Rank Percentile Minimization} is polynomial-time solvable as well. \\ 

Given a graph $G_0$ with $kt$ vertices for any $t \in \mathbb{N}$, if we let $G$ consist of $G_0$ with $k$ additional vertices forming $S^*$, then finding a rank of at least $t+1$ is the same as the {\sc NP}-complete problem $Vk${\sc (connected)}, so the uniform value case of {\sc Rank Maximization} is {\sc NP}-hard \cite{Vkconnectedresult, Vkconnected}. \\

We can also note that the uniform value case of {\sc Rank Maximization} is the same as the special case of the optimization version of $k$-{\sc Set Packing} where the set-system $(V,\mathcal{S})$ consists of the vertices of an undirected graph $G$ with its subsets of size $k$ that form a connected subgraph of $G$. In other words, given an undirected graph $G$, our objective is to exhibit a maximum-cardinality union of vertex-disjoint induced subgraphs of $G$ with $k$ vertices (henceforth referred to as \textit{$k$-components}). Moreover, we have a $\frac{k+2}{3}$-approximation algorithm for $k$-{\sc Set Packing} \cite{localsearchapprox}. \\

If $G$ is a connected circulant graph, then it is known that $G$ contains a Hamiltonian path \cite{hampathcirculant} and hence contains $\lfloor \frac{n}{k} \rfloor$ disjoint $k$-components, since we can place the first $k$ vertices in a $k$-component and recurse. We give an algorithmic proof of this fact, allowing us to compute a rank-maximizing partition of $V$ in polynomial-time. Denote $G = C_n(s_1,\dots,s_l)$, where $s_1,\dots,s_l < \frac{n}{2}$ are called \textit{jumps}; writing $V = \{0,1,\dots,n-1\}$, we have $(i,j) \in E \iff j-i \equiv \pm s_r$ (mod $n$).



 
\begin{theorem}
    The uniform value case of {\sc Rank Maximization} with the additional restriction that $G \backslash S^*$ be connected and circulant is polynomial-time solvable.  
\end{theorem}

\begin{proof}
    We prove the stronger statement that there exists a Hamiltonian path in $G \backslash S^*$ by strong induction on the number of jumps $l$. If $l = 1$, then $G \backslash S^*$ is connected if and only if $\gcd(n,s_1) = 1$, but this means that $\{s_1\}$ generates the additive group $\mathbb{Z}/n\mathbb{Z}$, and so the path $(0,s_1,2s_1,\dots,(n-1)s_1)$ is Hamiltonian. Now for fixed $l > 1$, assume the statement holds for $l-1,l-2,\dots,1$. Consider the following circulant subgraph of $G \backslash S^*$, $C_n(s_1)$. It consists of $g_1 = \frac{\mathrm{lcm}(n,s_1)}{s_1} = \gcd(n,s_1)$ cycles (the number of $s_1$ increments needed to return to the initial vertex), each having $b_1 = \frac{n}{g_1}$ vertices by symmetry. Also note that the vertices $0,1,\dots,g_1-1$ are in $g_1$ distinct cycles (because $n$ and $s_1$ are both multiples of $g_1$ and hence $ts_1 - n$ is divisible by $g_1$ for any $t \in \mathbb{Z}$, so no two elements of a cycle may have a difference of greater than $-g_1$ but less than $g_1$ (mod $n$)). Thus label the cycles as $0,1,\dots,g_1-1$ corresponding to which of these vertices it contains. Now in each cycle, perform edge contractions until each cycle only has one vertex, denoted $0,1,\dots,g_1-1$ as above. Note that the vertices in cycle $c$ are described by the set $\{c + ts_1$ (mod $n$)$: t \in \mathbb{Z}\}$, and so the set of possible differences between two vertices in cycles $c_1$ and $c_2$ (mod $n$) is given by $\{c_2-c_1 + ts_1$ (mod $n$)$: t \in \mathbb{Z}\} = -\{c_1-c_2 + ts_1$ (mod $n$)$: t \in \mathbb{Z}\}$. \\ 
    
    Fix $x \in \{2,\dots,l\}$ and consider $s_x$. Note that given two cycles $c_1,c_2 \in \{0,1,\dots,g_1-1\}$, $c_2-c_1 \in \{-\lfloor \frac{g_1}{2} \rfloor,-\lfloor \frac{g_1}{2} \rfloor +1,\dots,0,1,\dots,\lfloor \frac{g_1}{2} \rfloor\}$ (mod $n$). Now, if there exists $t_0 \in \mathbb{Z}$ such that $c_2-c_1 + t_0s_1 = \pm s_x$ (mod $n$), then $-(c_2-c_1) + -t_0s_1 = \mp s_x$ (mod $n$). Moreover, replacing $\pm t_0s_1$ with another integer (mod $n$) results in adding a non-zero multiple of $s_1$, hence a non-zero multiple of $g_1$, to the left-hand side of the equation (mod $n$), but $\pm (c_2-c_1)$ is uniquely determined by its remainder mod $n$. Therefore, a jump of $s_x$ either does not induce any edges in the contracted graph or induces an edge between $c_2$ and $c_1$ in the contracted graph exactly when $c_2-c_1 = \pm s_x'$ for unique $s_x' \in \{-\lfloor \frac{g_1}{2} \rfloor,-\lfloor \frac{g_1}{2} \rfloor +1,\dots,0,1,\dots,\lfloor \frac{g_1}{2} \rfloor\}$. \\
    
    Hence the contracted graph is circulant (and clearly it is connected) with at most $l-1$ distinct jumps, so by the inductive hypothesis there exists a Hamiltonian path in that graph. Moreover, the above expressions show that, if two vertices from different cycles are adjacent in $G \backslash S^*$, then by adding the multiples $s_1,2s_1,\dots,(b_1-1)s_1$ (mod $n$) to both of the vertices, we see that every vertex in one cycle has an edge to a vertex in the other cycle in $G \backslash S^*$. As a result, by replacing each vertex in the Hamiltonian path of the contracted graph with a Hamiltonian path in the corresponding cycle, we obtain a Hamiltonian path in $G \backslash S^*$ (since in $G \backslash S^*$, we can move between cycles that share an edge by taking an edge from any vertex in the first cycle).
\end{proof}

Furthermore, in the setting of the preceding theorem, unless there are no large subsets in $G$ (in which case the problem is trivial), Lemma \ref{max_approx} tells us to just focus on maximizing the number of large subsets to maximize the percentile. 

\begin{corollary}
    The uniform value case of {\sc Rank Percentile Maximization} with the additional restriction that $G \backslash S^*$ be connected and circulant is polynomial-time solvable.  
\end{corollary}

\section{Variants of the Problem} \label{variants}

\subsection{Equivalence Class Variant}

For some of the examples given above, it makes more sense to study the following \textit{equivalence class} variant of the above problems obtained by requiring that $E$ be an equivalence relation and redefining $\mathcal{S}$ as the set of singletons and equivalence classes of $V$. It turns out that, while there are still exponentially many possible partitions, this additional structure results in all four of the above problems becoming polynomial-time solvable.

\begin{corollary}
    The equivalence class variant of {\sc Rank Minimization} is polynomial-time solvable.
\end{corollary}

\begin{proof}
    Without the additional restrictions on $\mathcal{S}$, the algorithm in Theorem \ref{min_in_P} would just take as the partition sets the equivalence classes associated with each large element along with any remaining singletons. The minimum rank may only increase by removing subsets from $\mathcal{S}$, so this partition still achieves the minimum rank.  
\end{proof}

\begin{theorem}
    The equivalence class variant of {\sc Rank Maximization} is polynomial-time solvable.
\end{theorem}

\begin{proof}
    Consider the partition of $V$ obtained by taking each equivalence class exactly when at most one of its elements are large---excluding the equivalence class containing $S^*$ (for which the partition is predetermined). Then reverting any of these selected equivalence classes into singletons does not increase the rank of $S^*$ (since the value of a subset containing a large element does not increase), whereas consolidating singletons sets in a non-selected equivalence class decreases the rank (since by construction there are at least two large elements), and these are the only admissible deviations from this partition.   
\end{proof}

As for optimizing the percentile of $S^*$, respecting equivalence classes greatly simplifies the space of admissible partitions, decomposing it into linearly many choices: one for each equivalence class. 

\begin{theorem}
    The equivalence class variants of {\sc Rank Percentile Minimization} and {\sc Rank Percentile Maximization} are polynomial-time solvable. 
\end{theorem}

\begin{proof}
    We argue for {\sc Rank Percentile Minimization}; {\sc Rank Percentile Maximization} follows by a symmetric argument. Consider the following algorithm: initialize $P = P_0$ to be the partition consisting of singletons and $S^*$, and $p = p_0$ the percentile under $P$. While there exists an equivalence class not in $P$ for which $\frac{l+.5m-1}{l+m+s-1} > p$, add such an equivalence class to $P$ by combining the corresponding singletons and update $p$ to the percentile under $P$. \\
    
    Upon termination, we claim that $p$ is the minimum percentile $p^*$ of $S^*$ under a valid partition, and that it is achieved by $P$. To see this, note that a partition with percentile $p^*$ must contain any equivalence class for which $\frac{l+.5m-1}{l+m+s-1} > p^*$; otherwise, Lemma \ref{part_comb} implies that combining the equivalence class's singletons decreases the percentile, contradicting optimality of $p^*$. Any percentile $p$ realized during the execution of the algorithm necessarily satisfies $p \geq p^*$, so the equivalence classes it adds are part of an optimal partition $P^*$. To see that the rest are not, after the algorithm's termination we have that for any equivalence class not in $P$, $\frac{l+.5m-1}{l+m+s-1} \leq p$; if the inequality is strict, Lemma \ref{part_comb} states that adding the corresponding subset to $P$ would increase the percentile, and then the criterion for such only becomes weaker if such subsets have already been added by repeated application of the lemma. If it is equality, Lemma \ref{frac_avg} reveals adding the subset does not change the percentile and hence may be included or excluded. Thus $P$ consists of exactly the subsets in an optimal partition $P^*$, so the algorithm is correct. For runtime, we can sort, compare, and compute percentiles in polynomial time (e.g., by cross multiplication).       
\end{proof}

\subsection{Hierarchical Category Variant} \label{hierarchical_category_variant}

It is natural to extend the formalism in the preceding section to allow for subcategories of arbitrary depth in classifying elements and to allow an element to be classified under multiple (sub)categories, each instance having a different weight (value under $\mu$). In order to do this, we modify the set-up of our rank optimization problems as follows:

\begin{itemize}
    \item $G$ is comprised of the leaves of a rooted forest consisting of trees $T_1,\dots,T_m$ in which no nodes have exactly 1 child. 
    \item The root of each tree $T_j$ represents a \textit{category} of $V$ identified with all of its descendants that are leaves. Each descendant of the root that in turn has descendants is a \textit{subcategory} (of each of its ancestors) identified with all of its descendants that are leaves (we also still consider it a category of $V$). (In the case where the root is a leaf, we no longer call it a category of $V$, but it is still an element of $V$). 
    \item The elements of $V$ are partitioned into classes $\Tilde{V}$, each one representing all instances of a particular item under the category given by its parent node.
    \item The rank and percentile of $S^*$ are only optimized over valid partitions of $V$ such that each subset (including $S^*$) is a category of $V$ or the elements in a class in $\Tilde{V}$ that are not in a category subset in the partition.
\end{itemize}

Sometimes, we also add another node $v_0$ whose children are the roots of the above trees, yielding a composite tree $T'$. 

These rules are motivated by several real-world examples of ranking. \\ 

First, food manufacturers have numerous degrees of freedom when listing ingredients on packaged food labels. In the United States, for instance, even though the Food and Drug Administration (FDA) requires that all ingredients be listed in descending weight order, an ingredient consisting of multiple subingredients may either be listed on the label according to its total weight (with the subingredients listed in descending weight order using parentheses) or its subingredients may be listed independently without reference to the original ingredient (including recursive application of parentheses), cf. \textit{21 CFR 101.4(b)(2)} \cite{ingredientlabelrules}. One consequence of these rules is that an ingredient with no subingredients may be listed in multiple places on a food label according to its weight deriving from that specific place. For example, ``sugar'' may be listed as a subingredient of ``milk chocolate'' (its order among the subingredients determined by its weight in the milk chocolate) and as its own ingredient in the ingredient list (its order determined by its weight that is not part of the milk chocolate). \\

Similarly, there are many ways to organize the myriad greenhouse emissions that contribute to climate change, and the same type of emission can be broken up into the different uses that result in it being emitted. For this, the EPA uses the ``Common Reporting Tables" (CRT) prescribed by the Paris Agreement. As an example, ``Incineration of Waste" is part of the broader CRT Source Category 1A consisting of energy industry sources, but conversely can be decomposed further based on the gas emitted, e.g. $\mathrm{CO_2}$ (even if these gases are released in other contexts) \cite{EPAcategories}. \\   

As for the computational nature of this setting, this is the first setting where just minimizing rank is {\sc NP}-hard, resulting in all four of our problems being intractable and introducing the need for additional structure. 

\begin{theorem}
    The hierarchical category variant of {\sc Rank Minimization} is {\sc NP}-hard.  
\end{theorem}

\begin{proof}
    We reduce from {\sc X3C}. Given an instance with $n$ elements and $m$ subsets, we create a tree $T_m$ of depth 1 for each subset, where there is one leaf node for each element in the corresponding subset. $\Tilde{V}$ is defined by declaring all leaf nodes for a given element equivalent (across all the trees we create). Lastly, create a tree with one node $v^*$, which will be the only element in $S^*$. Set $\mu(S^*) = 5$ and $\mu(\{v\}) = 2, v \neq v^*$. Therefore, each category (subset in the {\sc X3C} instance) and each class in $\Tilde{V}$ have value 6 under $\mu$, so they are large subsets, but a subset of a class with just 2 of its representatives has value 4 under $\mu$, which is a small subset. We claim there exists an exact cover if and only if the minimum rank of $S^*$ is $\frac{n}{3} + 1$. For the forward direction, let each category pertaining to a subset in the exact cover be in the partition, with no other categories in the partition. The cover has size $\frac{n}{3}$, and the remaining elements form subsets of classes in $\Tilde{V}$ of size 2, hence not large. Thus this partition achieves rank $\frac{n}{3} + 1$. For the reverse direction, we note that there cannot be more than $\frac{n}{3}$ categories selected in the partition achieving minimum rank, but if there are $k$ fewer categories selected, the partition has at least $3k-k$ more large subsets than $\frac{n}{3}$, so we must have $\frac{n}{3}$ category subsets. But now, if they do not form an exact cover, then the leftover part of the class for an element not covered is a large subset in the partition, a contradiction.        
\end{proof}

\begin{theorem}
    The hierarchical category variant of {\sc Rank Maximization} is {\sc NP}-hard. 
\end{theorem}

\begin{proof}
     We reduce from {\sc X3C}. Given an instance with $n$ elements and $m$ subsets, we create a tree $T_m$ of depth 1 for each subset, where there is one leaf node for each element in the corresponding subset. $\Tilde{V}$ is defined by declaring all leaf nodes for a given element equivalent (across all the trees we create). Lastly, create a tree with one node $v^*$, which will be the only element in $S^*$. Set $\mu(S^*) = 3$ and $\mu(\{v\}) = 2, v \neq v^*$. Therefore, each category (subset in the {\sc X3C} instance) and each class in $\Tilde{V}$ have value 6 under $\mu$, so they are large subsets, but a subset of a class with just 1 of its representatives has value 2 under $\mu$, which is a small subset. We claim there exists an exact cover if and only if the maximum rank of $S^*$ is $\frac{4n}{3} + 1$. For the forward direction, we can choose each subset in the {\sc X3C} instance that is part of the exact cover, which keeps each element's class having 2 elements, resulting in $\frac{n}{3} + n$ large subsets in total. For the reverse direction, note that at least $\frac{n}{3}$ category subsets must be in the partition, but that if $k$ more are selected, then at least $3k$ elements are covered twice, so $\frac{4n}{3}$ large subsets is not possible. Thus the number is exactly $\frac{n}{3}$, and they must only cover each element once so that the corresponding classes have 2 representatives left and are therefore large. Hence these subsets form an exact cover.         
\end{proof}

Moreover, by observing that the number of small and large subsets is optimized if and only if an exact cover exists, these reductions establish that the hierarchical category variants of {\sc Rank Percentile Minimization} and {\sc Rank Percentile Maximization} are {\sc NP}-hard.  

\subsubsection{Bounded Degree and $|\Tilde{V}| = |V|$ Case}

The hardness results in this subsection arose due to the combinatorial complexity introduced by allowing representatives to be part of multiple subingredients with different siblings. We now show that if this is disallowed, i.e. $|\Tilde{V}| = |V|$ (meaning there is only one instance of each item), then the rank problems become easy and the percentile ones are polynomial-time solvable if the maximum degree of $G$ is bounded by a constant. \\

First, for rank minimization, the procedure is the following modified breadth-first search procedure on $T'$. Add the root to a list called explore, set rank equal to 1, and initialize a partition $P = {S^*}$. While there is a node in explore, pop such a node $v$ from explore and set a variable called distinct equal to false. If $v$ is not equal to $S^*$ and there is not a path from $v$ to $S^*$ that only goes away from the root (i.e. that can only move from a node to one of its child nodes), then set distinct equal to true. If distinct is true and the value of at least one leaf node that is a descendant of $v$ is at least the value of $S^*$ under $\mu$, update $P = P \cup {v}$ and increment rank. Else if distinct is true, put all leaf nodes that are descendants of $v$ into $P$ as singletons. Else if distinct is false, add all child nodes of $v$ to explore. To see correctness, note that for each node $v$ encountered in explore such that distinct is true, if one of its leaf nodes is medium or large, we add a subset consisting of all the leaf nodes in the subtree with $v$ as the root, and otherwise we leave each leaf node in the subtree as a singleton in the partition. If the former condition holds, then at least one subset from the subtree will be medium or large, and we ensure this lower bound is achieved; if the latter condition holds, none of the subsets from the subtree will be medium or large. \\  


For rank maximization, the procedure is as follows. For each category that does not intersect $S^*$, exclude it from the partition if and only if at least one of its children is medium or large. The resultant partition is well-defined because leaf nodes can always form singleton subsets in the partition, and a category that is included in the partition must be medium or large with no such children or small itself, so none of its descendants may be medium or large (due to additivity of $\mu$) and hence all will be excluded from the partition except for leaf nodes. Similarly, correctness follows from a simple inductive argument. \\

For rank percentile minimization and maximization, a challenge is caused by the fact that we do not know the optimal percentile (and hence whether excluding a category improves the percentile) a priori, and even if we did, the number of subsets we have added to any partition built iteratively will affect the magnitude of changes from subsequent subsets to the percentile. Therefore, there are several senses in which locally optimal decisions do not lead to a globally optimal solution. However, it is feasible to find the globally optimal solution if the degree of each category vertex is not too large; the hierarchical category variant of {\sc Rank Percentile Minimization} and {\sc Rank Percentile Maximization} are fixed-parameter tractable. 

\begin{theorem}
    The hierarchical category variant of {\sc Rank Percentile Minimization} and {\sc Rank Percentile Maximization} under the restriction that $|\Tilde{V}| = |V|$ are solvable in time $O\left(\left(\frac{e\cdot n}{\max\{{d,m\}}}\right)^{\max\{{d,m\}}}\right)$, where $d$ is the maximum tree degree of a vertex in $G$. 
\end{theorem}

\begin{proof}
    Let $\mathrm{OPT}[u,c]$ be the minimum percentile when partitioning the elements of category $u$ into $c$ subsets (if no such partition exists, it is instead equal to $\infty$). Suppose $u$ has subcategories $u_1,\dots,u_r$ as children (excluding $S^*$ as a subcategory) and $c > 1$. Then the following recurrence holds: 
    \[
    \mathrm{OPT}[u,c] =
    \min_{\substack{c_1 \geq 1, \dots, c_r \geq 1 \\ c_1 + \cdots + c_r = c}} 
    \left\{ \frac{c_1+1}{c+1}\mathrm{OPT}[u_1,c_1] + \cdots + \frac{c_r+1}{c+1}\mathrm{OPT}[u_r,c_r] - \frac{r-1}{2(c+1)}
    \right\}
    \]
     To see this, note that there are two options for category $u$ in a partition of $V$: it may be one subset in the partition (corresponding to the below case), or each subcategory of $u$ may be one subset in the partition (or subdivided into multiple subsets according to the above rules). Therefore, we minimize over all partitions considered in the hierarchical category variant. The coefficients in the recurrence are selected to give the percentile of the partition obtained by combining the choice of subsets for each subcategory of $u$. If $u$ is a leaf node or $c = 1$, we instead have $$\mathrm{OPT}[u,c] = 
    \begin{cases}
    \frac{1.5}{2} & \text{if } u \text{ is large}, u \not\in S^*, \text{and } c = 1 \\ 
    \frac{1}{2} & \text{if } u \text{ is medium}, u \not\in S^*, \text{and } c = 1 \\
    \frac{.5}{2} & \text{if } u \text{ is small}, u \not\in S^*, \text{and } c = 1 \\
    \infty & \text{otherwise}
    \end{cases}.$$ 
    As for computation, each choice of $(c_1,\dots,c_r)$ in the second $\min$ operator requires $O(1)$ runtime, and there are $\binom{c-1}{r-1}$ such choices by a stars-and-bars argument. Repeating for at most $n$ categories and using that $\binom{c-1}{r-1} \leq (\frac{e\cdot(c-1))}{r-1})^{r-1}$ with $c \leq n$ and $r \leq \max\{d,m\}$, we get the claimed runtime. We return $$\min_{1 \leq c \leq n}\mathrm{OPT}[v_0,c].$$ \\

    To solve the hierarchical category variant of {\sc Rank Percentile Maximization}, replace $\infty$ with $-\infty$ and $\min$ with $\max$ in the above algorithm; the proof of correctness is analogous. 
\end{proof}

\subsection{Grid Variants}

An interesting modification of the preceding, where $G$ has no natural embedding in $\mathbb{R}^n$, is to view the graph $G$ as a two-dimensional grid, perhaps representing locations or equally spaced points on a map or a discretization thereof. We also allow a set of vacancies $V_-$ in the grid, representing points that would not be actual locations within $G$ (however interpreted). Therefore, without loss of generality we may regard $V$ as $[1,l] \times [1,w] \backslash V_-$, where $V_- \subset [1,l] \times [1,w]$ and $n = |V| = lw - |V_-|$. Lastly, two points $(a,c),(b,d) \in V$ are similar if they occupy adjacent positions in the grid: $((a,c),(b,d)) \in E \iff |b-a| + |d-c| = 1$. \\ 

For instance, if $G$ represents the United States as projected on a map, then we may leave some vacancies for Canada and the Atlantic Ocean instead of, say, stacking New England on top of New York, Pennsylvania, and Ohio. \\ 

Denote $\square_{abcd} := [a,b] \times [c,d]$, the rectangular grid spanning $x$-coordinates $a$ to $b$ and $y$-coordinates $c$ to $d$ (inclusive). \\

In the below variants of the problem, we study partitions of the grid $[1,l] \times [1,w]$ into axis-parallel rectangles. By taking the intersection of each rectangle with $V$, we obtain an induced partition of $V$. As usual, we seek to optimize the percentile of $S^*$ (a subgrid of $V$ that without loss of generality does not contain any vacancies) under a valid partition of $V$; we just place restrictions on which valid partitions may be considered, as described for each of the proceeding variants below. \\    

For both variants, we specify a collection of subsets of $V_-$ representing valid subsets of vacancies to place in a rectangle with elements of $V$. This may be specified in polynomial space by giving a list of rectangles $\mathcal{R} \subset \mathcal{P}(\{\square_{abcd}: 1 \leq a \leq b \leq l, 1 \leq c \leq d \leq w\})$ that may be included in the partition of $V$, with the caveat that if two rectangles contain elements of $V$ and the same subset of $V_-$, they are either both allowed or both disallowed (and rectangles comprised only of elements of $V$ are allowed). In other words, $\square_{abcd} \cap V_- = \emptyset \implies \square_{abcd} \in \mathcal{R}$ and $\square_{abcd} \cap V_- = \square_{a'b'c'd'} \cap V_- \implies (\square_{abcd} \in \mathcal{R} \iff \square_{a'b'c'd'} \in \mathcal{R})$.  


\subsubsection{Grid with Hierarchical Rectangles Variant}

In the \textit{grid with hierarchical rectangles} variant, we consider the restriction where the only admissible partitions of $V$ are those induced by hierarchically partitioning the grid into rectangles with axis-parallel edges. That is, in each step, we take a subgrid from a previous step ($[1,l] \times [1,w]$ for the first step) and draw a horizontal or vertical line between two rows or columns thereof, respectively, resulting in the cut subgrid becoming two subgrids separated by the line instead. Upon termination, each subgrid containing elements of $V$ must be in $\mathcal{R}$, which induces a partition of $V$.    




\begin{theorem}
    The grid with hierarchical rectangles variant of {\sc Rank Minimization} and {\sc Rank Maximization} are polynomial-time solvable.
\end{theorem}
    
\begin{proof}

Consider the following algorithm.

\begin{algorithm}[H]
\caption{Given an instance $(V,E,\mu,S^*)$ of {\sc Rank Minimization} in the grid with hierarchical rectangles variant, compute the minimum rank of $S^*$}\label{hierarchical_rank}
\begin{algorithmic}

\State Initialize four-dimensional table $\mathrm{OPT}$ of size $l \times l \times w \times w$ with {\sc null} entries

\State \Call{Compute-Opt}{$1,l,1,w$}
\newline
\Procedure{Compute-Opt}{$a,b,c,d$}
\If{$\mathrm{OPT}[a,b,c,d]$ is not {\sc null}}
    \State \Return $\mathrm{OPT}[a,b,c,d]$

\ElsIf{$\square_{abcd} = S^*$ \text{ or } $\square_{abcd} \subset V_-$}
    \State $\mathrm{OPT}[a,b,c,d] = 0$

\ElsIf{($\square_{abcd} \cap S^* \neq \emptyset$ \text{ and } $\square_{abcd}^\mathrm{c} \cap S^* \neq \emptyset$)}
    \State $\mathrm{OPT}[a,b,c,d] = \infty$
\EndIf

\State $c_b \gets \begin{cases}
1 & \text{if } \square_{abcd} \text{ is large, not } S^* \subset \square_{abcd} \text{, and } \square_{abcd} \in \mathcal{R} \\
0 & \text{if } \square_{abcd} \text{ is medium or small, not } S^* \subset \square_{abcd} \text{, and } \square_{abcd} \in \mathcal{R} \\
\infty & \text{otherwise}
\end{cases}$

\State $c_v \gets \infty$
\State $c_h \gets \infty$

\If{$a < b$}
    \State $c_v \gets \min_{a \leq i < b}$ \{\Call{Compute-Opt}{$a,i,c,d$} + \Call{Compute-Opt}{$i+1,b,c,d$}\}
\EndIf

\If{$c < d$}
    \State $c_h \gets \min_{c \leq j < d}$ \{\Call{Compute-Opt}{$a,b,c,j$} + \Call{Compute-Opt}{$a,b,j+1,d$}\}
\EndIf

\State $\mathrm{OPT}[a,b,c,d] = \min\{c_b,c_v,c_h\}$
 
\State \Return $\mathrm{OPT}[a,b,c,d]$
\EndProcedure        

\end{algorithmic}
\end{algorithm}

We claim that for all $1 \leq a \leq b \leq l, 1 \leq c \leq d \leq w$, {\sc Compute-Opt}($a,b,c,d$) gives the minimum number of large subsets of $\square_{abcd} \backslash V_-$ when $\square_{abcd} \backslash V_-$ is hierarchically partitioned into subsets in $\mathcal{R}$ with the further restriction that $S^*$ is one of the subsets (or that $S^*$ does not intersect $\square_{abcd}$ and its complement, in which case {\sc Compute-Opt}($a,b,c,d$) should be infinity). Taking $a = 1, b = l, c = 1, d = w$ will then prove the correctness of the algorithm. \\ 

First, if $\square_{abcd} = S^*$ or $\square_{abcd} \subset V_-$, then the subgrid we are partitioning is vacuous and hence there are 0 large subsets. Similarly, if $S^*$ intersects $\square_{abcd} \backslash V_-$ but also its complement, $S^*$ may not be one of the subsets in the former and the algorithm must assign a value of infinity to {\sc Compute-Opt}($a,b,c,d$). \\ 

Otherwise, $\square_{abcd}$ does not intersect $S^*$ or it properly contains $S^*$, and in both cases it includes elements of $V$ that need to be partitioned. There are two options. One, we may leave $\square_{abcd}$ as one subset, but only if it is in $\mathcal{R}$ and does not contain $S^*$. In that case, the number of large subsets is 1 if $\square_{abcd}$ is large and 0 otherwise (if it is not possible, $\square_{abcd}$ is not a singleton and hence the second option will be possible). Two, we may subdivide $\square_{abcd}$ by cutting it vertically (assuming $a \neq b$) in one of $b-a$ positions or horizontally (assuming $c \neq d$) in one of $d-c$ positions and only further partitioning each of these subdivisions (no subsets crossing this cut). In that case, the minimum number of large subsets is the minimum sum (over all positions to make the cut) of the minimum number of large subsets for each resultant subgrid, since a hierarchical partition of $\square_{abcd}$ with a horizontal or vertical cut across which there may be no subsets is exactly comprised of a hierarchical partition of each subgrid. Finally, the minimum number of large subsets of $\square_{abcd} \backslash V_-$ is the minimum number resulting from each option. Hence if we suppose that {\sc Compute-Opt}($a',b',c',d'$) has the desired behavior for all $\square_{a'b'c'd'} \subsetneq \square_{abcd}$, the claim follows by induction on the size of the subgrid $(b-a)(d-c)$ (if it reaches 1, then no recursive calls are made). \\ 



To solve the grid with hierarchical rectangles variant of {\sc Rank Maximization}, we simply replace $\infty$ with $-\infty$, $\min$ with $\max$, ``large" with ``large or medium", and ``medium or small" with ``small" in Algorithm \ref{hierarchical_rank}; the proof of correctness is analogous. 
\end{proof}

\begin{theorem}
    The grid with hierarchical rectangles variant of {\sc Rank Percentile Minimization} and {\sc Rank Percentile Maximization} are polynomial-time solvable.
\end{theorem}

\begin{proof}

Consider the following algorithm.

\begin{algorithm}[H]
\caption{Given an instance $(V,E,\mu,S^*)$ of {\sc Rank Percentile Minimization} in the grid with hierarchical rectangles variant, compute the minimum percentile of $S^*$}\label{hierarchical_percentile}
\begin{algorithmic}
\State Initialize five-dimensional table $\mathrm{OPT}$ of size $l \times l \times w \times w \times n$ with {\sc null} entries
\State \Return $\min_{1 \leq k^* \leq n}$ \{\Call{Compute-Opt}{$1,m,1,n,k^*$}\}
\newline
\Procedure{Compute-Opt}{$a,b,c,d,k$}
\If{$\mathrm{OPT}[a,b,c,d,k]$ is not {\sc null}}
    \State \Return $\mathrm{OPT}[a,b,c,d,k]$

\ElsIf{($\square_{abcd} = S^*$ \text{ or } $\square_{abcd} \subset V_-$) \text{ and } $k = 0$}
    \State $\mathrm{OPT}[a,b,c,d,k] = 0$

\ElsIf{($\square_{abcd} \cap S^* \neq \emptyset$ \text{ and } $\square_{abcd}^\mathrm{c} \cap S^* \neq \emptyset$) \text{ or } $k = 0$}
    \State $\mathrm{OPT}[a,b,c,d,k] = \infty$

\EndIf

\State $c_b \gets \begin{cases}
\frac{1.5}{2} & \text{if } \square_{abcd} \text{ is large, not } S^* \subsetneq \square_{abcd}, \square_{abcd} \in \mathcal{R}, \text{ and } k = 1 \\ 
\frac{1}{2} & \text{if } \square_{abcd} \text{ is medium, not } S^* \subsetneq \square_{abcd}, \square_{abcd} \in \mathcal{R}, \text{ and } k = 1 \\
\frac{.5}{2} & \text{if } \square_{abcd} \text{ is small, not } S^* \subsetneq \square_{abcd}, \square_{abcd} \in \mathcal{R}, \text{ and } k = 1 \\
\infty & \text{otherwise}
\end{cases}$

\State $c_v \gets \infty$
\State $c_h \gets \infty$

\If{$a < b$}
    \State $c_v \gets \frac{1}{k+1}\min_{a \leq i < b, 0 \leq k_h \leq k}$ \{$(k_h+1)$\Call{Compute-Opt}{$a,i,c,d,k_h$} + $(k-k_h+1)$\Call{Compute-Opt}{$i+1,b,c,d,k-k_h$} $ - \frac{1}{2}$\}
\EndIf

\If{$c < d$}
    \State $c_h \gets \frac{1}{k+1}\min_{c \leq j < d, 0 \leq k_v \leq k}$ \{$(k_v+1)$\Call{Compute-Opt}{$a,b,c,j,k_v$} + $(k-k_v+1)$\Call{Compute-Opt}{$a,b,j+1,d,k-k_v$} $ - \frac{1}{2}$\}
\EndIf

\State $\mathrm{OPT}[a,b,c,d,k] = \min\{c_b,c_v,c_h\}$
 
\State \Return $\mathrm{OPT}[a,b,c,d,k]$
\EndProcedure        

\end{algorithmic}
\end{algorithm}

The proof of correctness is similar to that for Algorithm \ref{hierarchical_rank}, except that we are now minimizing percentile instead of rank. Accordingly, we claim that {\sc Compute-Opt}$(a,b,c,d,k)$ gives the minimum percentile of $S^*$ when $\square_{abcd} \backslash V_-$ is hierarchically partitioned into exactly $k \geq 0$ subsets in $\mathcal{R}$ other than $S^*$ with the further restriction that $S^*$ be one of the subsets (or that $S^*$ not intersect $\square_{abcd}$ and its complement, in which case {\sc Compute-Opt}($a,b,c,d,k$) should be infinity). Moreover, if there is no such partition with our value of $k$, {\sc Compute-Opt}($a,b,c,d,k$) should be infinity. We also adopt the convention that if $k = 0$ and there are no non-$S^*$ elements of $V$ in $\square_{abcd}$, {\sc Compute-Opt}($a,b,c,d,k$) should be infinity. We call this the \textit{minimum percentile of $S^*$ for $\square_{abcd}$}. Taking $a = 1, b = l, c = 1, d = w$ for each $1 \leq k^* \leq n$ will then prove the correctness of the algorithm, since the minimum percentile of $S^*$ in the original instance is achieved by the minimum percentile of $S^*$ under a valid partition of $V$ with the number of subsets other than $S^*$ being between 1 and $n$, so the algorithm returns the minimum percentile of $S^*$ over all valid partitions of $V$. \\

First, if $\square_{abcd} = S^*$ or $\square_{abcd} \subset V_-$, then the subgrid we are partitioning is vacuous and we must have $k = 0$ for such a partition. Similarly, if $S^*$ intersects  $\square_{abcd} \backslash V_-$ but also its complement, $S^*$ may not be one of the subsets in the former and the algorithm must assign a value of infinity to {\sc Compute-Opt}($a,b,c,d,k$). Lastly, if $k = 0$ but there are non-$S^*$ elements of $V$ in $\square_{abcd}$, then we cannot put those element in any partition subset and the algorithm must assign a value of infinity to {\sc Compute-Opt}($a,b,c,d,k$). \\ 

Otherwise, $\square_{abcd}$ does not intersect $S^*$ or it properly contains $S^*$, and in both cases it includes elements of $V$ that need to be partitioned. We have also ensured at this point that $k \geq 1$. There are two options. One, we may leave $\square_{abcd}$ as one subset, but only if it is in $\mathcal{R}$ and does not contain $S^*$, and $k = 1$. In that case, by definition the percentile of $S^*$ is $\frac{1.5}{2}$ if $\square_{abcd}$ is large, $\frac{1}{2}$ if medium, and $\frac{.5}{2}$ if small otherwise (if it is not possible, $\square_{abcd}$ is not a singleton and hence the second option will be possible). Two, we may subdivide $\square_{abcd}$ by cutting it vertically (assuming $a \neq b$) in one of $b-a$ positions or horizontally (assuming $c \neq d$) in one of $d-c$ positions and only further partitioning each of these subdivisions (no subsets crossing this cut). In that case, by direct computation the minimum percentile of $S^*$ for $\square_{abcd}$ is the minimum (over all positions to make the cut and possible number of subsets in a partition of each adding up to $k$) convex combination computed by the algorithm of the minimum percentile of $S^*$ for each resultant subgrid, since a hierarchical partition of $\square_{abcd}$ with a horizontal or vertical cut across which there may be no subsets is exactly comprised of a hierarchical partition of each subgrid. Finally, the minimum percentile of $S^*$ for $\square_{abcd}$ is the minimum percentile resulting from each option. Hence if we suppose that {\sc Compute-Opt}($a',b',c',d',k'$) has the desired behavior for all $\square_{a'b'c'd'} \subsetneq \square_{abcd}$ and $k' \leq k$, the claim follows by induction on the size of the subgrid $(b-a)(d-c)$ and $k$ (if the variable reaches 1 or 0, respectively, then no recursive calls are made). \\  

To solve the grid with hierarchical rectangles variant of {\sc Rank Percentile Maximization}, we simply replace $\infty$ with $-\infty$ and $\min$ with $\max$ in Algorithm \ref{hierarchical_percentile}; the proof of correctness is analogous. 
\end{proof}


\subsubsection{Grid with Rectangles Variant}

In the \textit{grid with rectangles} variant, we consider the restriction where the only admissible partitions of $V$ are those induced by partitioning the grid into rectangles with axis-parallel edges, not just ones arising from hierarchical partitioning. Additionally, each subset in the induced partition of $V$ must be in $\mathcal{R}$.

\begin{theorem}
    The grid with rectangles variant of {\sc Rank Maximization} is {\sc NP}-hard. 
\end{theorem}

\begin{proof}
    The special case of {\sc Maximum Disjoint Set} where the collection of objects is comprised of axis-parallel rectangles in $\mathbb{R}^2$ is NP-hard \cite{mdshardness}, and this remains true if each rectangle consists of vertices with integer coordinates \cite{integralrectanglehardness}, so we fix such an instance of {\sc Maximum Disjoint Set}. To convert this into an instance of {\sc Rank Maximization} in the grid with rectangles variant, we proceed as follows. 

    \begin{itemize}
        \item The grid dimensions are 2 times the maximum rectangle coordinate in the corresponding dimension
        \item For all $v \in V$, $\mu(\{v\}) = \frac{3}{5}\mu(S^*)$
        \item The vacancies are all grid points that have an odd coordinate
        \item $\mathcal{R}$ consists of the doubled version of each rectangle---i.e., the rectangle obtained by multiplying each of its vertex coordinates by 2---as delineated by the subset of vacancies inside its boundary
    \end{itemize}

Clearly, two rectangles intersect if and only if their doubled versions intersect. By choice of $\mu$, any rectangular subset $S \in \mathcal{R}$ of the doubled grid that contains at least two elements of $V$ satisfies $\mu(S) \geq \mu(S^*)$, and any other subset does not. Hence, maximizing the rank of $S^*$ in the {\sc Rank Maximization} instance is equivalent to choosing a collection of disjoint rectangular subsets of maximum cardinality. To show that we only consider subsets identified as rectangles in the {\sc Maximum Disjoint Set} instance, note that between any two grid points on the boundary of a rectangle in the {\sc Maximum Disjoint Set} instance, a vacancy is placed between these points in its doubled version. As a result, the doubled version of each non-degenerate, rectangular subset in the original grid uniquely contains the subset of vacancies inside its boundary.    
\end{proof}

Fortunately, the analogue with the better studied {\sc Maximum Disjoint Set} yields an approximation algorithm for our problem.  

\begin{theorem}
    For any $\epsilon > 0$, the grid with rectangles variant of {\sc Rank Maximization} has a polynomial-time approximation algorithm with approximation guarantee $\mathrm{OPT} \leq (2+\epsilon)\mathrm{ALG} - (1+\epsilon)$.
\end{theorem}

To see this, we leverage the fact that there is a $(2+\epsilon)$-approximation algorithm for the above special case of {\sc Maximum Disjoint Set} \cite{mdsapprox}. The same guarantees apply to maximizing the number of medium or large subsets in the grid with rectangles variant of {\sc Rank Maximization} because, conversely, this variant is a special case of {\sc Maximum Disjoint Set}; the vacancies and $\mathcal{R}$ effectively just restrict the collection of rectangles we may consider. \\

We also have the following relationship between rank in these two variants, since a hierarchical partition is just a special class of partition. 

\begin{corollary}
    Given an input $(G,V_-,\mu,S^*)$ for a grid variant problem, let $\mathrm{max \; rank}$ be the maximum rank in the grid with rectangles variant and let $\mathrm{max \; hierarchical \; rank}$ be the maximum rank in the grid with hierarchical rectangles variant; define $\mathrm{min \; rank}$ and $\mathrm{min \; hierarchical \; rank}$ analogously. Then $\mathrm{min \; rank} \leq \mathrm{min \; hierarchical \; rank} \leq \mathrm{max \; hierarchical \; rank} \leq \mathrm{max \; rank}$. 
\end{corollary}

\subsection{Applications of Recursive Algorithms for Rank and Percentile}

The above recursive algorithms for the grid variants of the problems allow us to easily modify the criteria for subsets in a valid partitions and the convex combination of ranks or percentiles that we compute. We consider two such applications below



\subsubsection{Numerical Grading over Multiple Periods}

In several academic and professional contexts, especially secondary schools in the United States, it is customary to evaluate people by recording a measure of average performance for a set number of terms and then computing a (possibly weighted) average of those marks. We consider an ex-post change in these terms to discuss the robustness of such grading practices.  

\begin{definition}
    The {\sc Weighted Average Maximization} problem is the following modification of the linear component case of {\sc Rank Percentile Maximization}: we now require  $S^* = \emptyset$. We are also given a measure $\mu^*$ such that $\mu \leq \mu^*$, and we drop the requirement that $\mu(S) = 0 \implies S = \emptyset$ but force $\mu^*(V) > 0$. The objective is now to select a valid partition $P = (S_1,S_2,\dots,S_c)$, where $S_1 = [1,i_1], S_2 = [i_1+1,i_2],\dots, S_c = [i_{c-1}+1,i_c]$, $i_1 < i_2 < \cdots < i_{c-1} < i_c = l$, and $\mu^*(S_1) > 0, \mu^*(S_2) > 0,\dots, \mu^*(S_c) > 0$ such that the grade $\frac{1}{i_1}\frac{\mu(S_1)}{\mu^*(S_1)} + \frac{1}{i_2-i_1}\frac{\mu(S_2)}{\mu^*(S_2)} + \cdots + \frac{1}{i_c-i_{c-1}}\frac{\mu(S_c)}{\mu^*(S_c)}$ is maximized.   
\end{definition}


In this formulation, the horizontal dimension represents time, and each vertex represents a discrete time period, such as one day. $\mu$ is the number of points earned for a particular interval of time, and $\mu^*$ is the maximum possible points for a particular interval of time. The goal is to partition the entire interval of time into marking periods such that the average percentage of points earned (weighted by the amount of time in each marking period) is as large as possible, yet manipulation by making some periods have zero points possible, effectively redistributing grading weight to other periods (including non-adjacent ones), is disallowed. \\

Because time is the only dimension of interest, we can obtain the below polynomial-time algorithm to solve this version by adapting the procedure for the linear component case of {\sc Rank Percentile Maximization}.

\begin{algorithm}[H]
\caption{Given an instance $(V,E,\mu^*,\mu,S^*=\emptyset)$ of {\sc Weighted Average Maximization}, compute the maximum grade of a partition}\label{numerical_grading}
\begin{algorithmic}

\For{$1 \leq i \leq n$}
    \For{$i \leq j \leq n$}
       \State $C[i,j] = 
        \begin{cases}
            \frac{1}{j-i+1}\frac{\mu(\{v_i,\dots,v_j\})}{\mu^*(\{v_i,\dots,v_j\})} & \text{if } \mu^*(\{v_i,\dots,v_j\}) > 0 \\
            -\infty & \text{otherwise} 
       \end{cases}$
    \EndFor
\EndFor

\State $\mathrm{OPT}[1] = C[1,1]$

\For{$1 < j \leq n$}
    \State $\mathrm{OPT}[j] = \sup_{1 \leq i < j} \{\mathrm{OPT}[i] + C[i+1,j]\}$    
\EndFor

\State \Return $\mathrm{OPT}[n]$
        
\end{algorithmic}
\end{algorithm}

\subsection{Gerrymandering with Hierarchical Rectangles}

There are two main impediments to using the preceding techniques to study gerrymandering: monotonicity of unsigned measures is frequently leveraged and subsets in a valid partition are not required to be a similar size (even if this outcome is likely in the maximization problems). We circumvent both by adapting the algorithm for {\sc Rank Maximization} in the grid with hierarchical rectangles variant, restricting to districts of this form in order to check the new redistricting conditions in our recursive solution approach. \\

In this section, we augment the general set-up of our four partitioning problems in Section \ref{general_case}, which required a vertex set $V$, a binary relation $E$ on $V \times V$ (which determines a set of connected vertices $\mathcal{S}$), a measure $\mu$ on $2^V$, and a subset $S^*$ of $V$. Now, we are also given $N \in \mathbb{N}$, a signed measure $\mu_R$ on $2^V$, and a parameter $\rho \in [0,\frac{1}{N+1})$ such that $|\mu_R| \leq \mu$, each valid partition must also have exactly $N$ subsets, and any subset $S \in \mathcal{S}$ in a valid partition must also satisfy $(1-\rho)\frac{\mu(V)}{N} \leq \mu(S) \leq (1+\rho)\frac{\mu(V)}{N}$. \\

We may now interpret the parameters as follows. A state, which is subdivided into the elements of $V$ (which we call \textit{precincts}) and whose adjacency is represented by $E$, is allocated $N$ districts in a legislative body, each of which should be contiguous and have population (measured by $\mu$) within a fraction $\rho$ of the eventual average district population for that state. Two parties are ascendant in the political system, labeled Player 1 and Player 2, and $\mu_R$ is the number of Player 1-leaning voters minus the number of Player 2-leaning voters (or the expectation thereof) for a given area, so by this convention the objective for this version of {\sc Rank Maximization} would be to choose a redistricting plan satisfying the above rules such that the number of districts that lean toward Player 1 by at least the margin of $S^*$ is maximized. For concreteness, we study the special case where $S^* = \emptyset$; we call this variant of rank maximization {\sc Gerrymandering with Contiguous Districts}. \\

More formally, we have the following definitions. 

\begin{definition}[Valid Contiguous Redistricting, District, and Slate]
    Given a tuple $(V,E,\mu,\mu_R,\rho,N)$ as above, we say a partition $(S_1,\dots,S_N)$ of $V$ is a valid contiguous redistricting if $S_1,\dots,S_N \in \mathcal{S}$ and $(1-\rho)\frac{\mu(V)}{N} \leq \mu(S_i) \leq (1+\rho)\frac{\mu(V)}{N}$ for all $i \in \{1,\dots,N\}$, and each subset $S_i$ is called a district thereof. Given a valid contiguous redistricting $(S_1,\dots,S_N)$ and $j \in \{1,2\}$, we say Player 1's slate $N_1$ is the number of $S_i$ for which $\mu_R(S_i) \geq 0$ and Player 2's slate $N_2$ is the number of $S_i$ for which $\mu_R(S_i) < 0$ (i.e., ties are broken in favor of Player 1). 
\end{definition}

\begin{definition}[Gerrymandering with Contiguous Districts]
    The {\sc Gerrymandering with Contiguous Districts} problem is as follows: given a tuple $(V,E,\mu,\mu_R,\rho,N)$ as above, select a valid contiguous redistricting $(S_1,\dots,S_N)$ that maximizes the number of $S_i$ for which $\mu_R(S_i) \geq 0$. 
\end{definition}

\begin{definition}[Valid Hierarchical Rectangle Redistricting]
    Given a tuple $(V,V_-,E,\mathcal{R},\mu,\mu_R,\rho,N)$ as above, we say a partition $(S_1,\dots,S_N)$ of $V$ is a valid hierarchical rectangle redistricting if it is a valid rectangle redistricting and is induced by hierarchically partitioning the grid into rectangles with axis-parallel edges. That is, in each step, we take a subgrid from a previous step ($[1,l] \times [1,w]$ for the first step) and draw a horizontal or vertical line between two rows or columns thereof, respectively, resulting in the cut subgrid becoming two subgrids separated by the line instead. 
\end{definition}

\begin{definition}[Gerrymandering with Hierarchical Rectangles]
    The {\sc Gerrymandering with Hierarchical Rectangles} problem is as follows: given a tuple $(V,V_-,E,\mathcal{R},\mu,\mu_R,\rho,N)$ as above, select a valid hierarchical rectangle redistricting $(S_1,\dots,S_N)$ that maximizes the number of $S_i$ for which $\mu_R(S_i) \geq 0$. 
\end{definition}

A polynomial-time algorithm solving {\sc Gerrymandering with Hierarchical Rectangles} by modifying the procedure for the grid with hierarchical rectangles variant of {\sc Rank Maximization} is presented below. 

\begin{theorem}
    {\sc Gerrymandering with Hierarchical Rectangles} is polynomial-time solvable.
\end{theorem}

\begin{proof}

Consider the following algorithm. 

\begin{algorithm}[H]
\caption{Given an instance $(V,V_-,E,\mathcal{R},\mu,\mu_R,\rho,N)$ of {\sc Gerrymandering with Hierarchical Rectangles}, compute the maximum slate of Player 1}\label{gerrymandering}
\begin{algorithmic}

\State Initialize four-dimensional table $\mathrm{OPT}$ of size $l \times l \times w \times w$ with {\sc null} entries

\State \Call{Compute-Opt}{$1,l,1,w$}
\newline
\Procedure{Compute-Opt}{$a,b,c,d$}
\If{$\mathrm{OPT}[a,b,c,d]$ is not {\sc null}}
    \State \Return $\mathrm{OPT}[a,b,c,d]$

\ElsIf{$\square_{abcd} \subset V_-$}
    \State $\mathrm{OPT}[a,b,c,d] = 0$

\EndIf

\State $c_b \gets \begin{cases}
1 & \text{if } \mu_R(\square_{abcd}) \geq 0 \text{, } (1-\rho)\frac{\mu(V)}{N} \leq \mu(\square_{abcd}) \leq (1+\rho)\frac{\mu(V)}{N} \text{, and } \square_{abcd} \in \mathcal{R} \\
0 & \text{if } \mu_R(\square_{abcd}) < 0 \text{, } (1-\rho)\frac{\mu(V)}{N} \leq \mu(\square_{abcd}) \leq (1+\rho)\frac{\mu(V)}{N} \text{, and } \square_{abcd} \in \mathcal{R} \\
-\infty & \text{otherwise}
\end{cases}$

\State $c_v \gets -\infty$
\State $c_h \gets -\infty$

\If{$a < b$}
    \State $c_v \gets \max_{a \leq i < b}$ \{\Call{Compute-Opt}{$a,i,c,d$} + \Call{Compute-Opt}{$i+1,b,c,d$}\}
\EndIf

\If{$c < d$}
    \State $c_h \gets \max_{c \leq j < d}$ \{\Call{Compute-Opt}{$a,b,c,j$} + \Call{Compute-Opt}{$a,b,j+1,d$}\}
\EndIf

\State $\mathrm{OPT}[a,b,c,d] = \max\{c_b,c_v,c_h\}$
 
\State \Return $\mathrm{OPT}[a,b,c,d]$
\EndProcedure        

\end{algorithmic}
\end{algorithm}

We claim that for all $1 \leq a \leq b \leq l, 1 \leq c \leq d \leq w$, {\sc Compute-Opt}($a,b,c,d$) gives the maximum number of rectangles with $\mu_R$ non-negative when $\square_{abcd} \backslash V_-$ is hierarchically partitioned into districts $D$ in $\mathcal{R}$ with the further restriction that for each $D$, $(1-\rho)\frac{\mu(V)}{N} \leq \mu(D) \leq (1+\rho)\frac{\mu(V)}{N}$ (unless such a partition is not possible, in which case {\sc Compute-Opt}($a,b,c,d$) should be negative infinity). Taking $a = 1, b = l, c = 1, d = w$ will then prove the correctness of the algorithm because if $V$ is hierarchically partitioned in this manner, by construction the result is a valid hierarchical rectangle redistricting with exactly $N$ districts. Indeed, since $\rho < \frac{1}{N+1}$, $N_- < N$ districts have population at most $(1+\rho)\frac{\mu(V)}{N}N_- < \mu(V)$ and $N_+ > N$ districts have population at least $(1-\rho)\frac{\mu(V)}{N}N_+ > \mu(V)$. \\ 

First, if $\square_{abcd} \subset V_-$, then the subgrid we are partitioning is vacuous and hence there are 0 districts satisfying any claimed conditions. \\

Otherwise, $\square_{abcd}$ includes elements of $V$ that need to be partitioned. There are two options. One, we may leave $\square_{abcd}$ as one district, but only if it is in $\mathcal{R}$ and $(1-\rho)\frac{\mu(V)}{N} \leq \mu(\square_{abcd}) \leq (1+\rho)\frac{\mu(V)}{N}$. In that case, the number of districts contributing to Player 1's slate is 1 if $\square_{abcd}$ satisfies $\mu_R(\square_{abcd}) \geq 0$ and 0 otherwise (if it is not possible, $\square_{abcd}$ is not a singleton and hence the second option will be possible). Two, we may subdivide $\square_{abcd}$ by cutting it vertically (assuming $a \neq b$) in one of $b-a$ positions or horizontally (assuming $c \neq d$) in one of $d-c$ positions and only further partitioning each of these subdivisions (no districts crossing this cut). In that case, the maximum number of districts contributing to Player 1's slate is the maximum sum (over all positions to make the cut) of the maximum number of such districts for each resultant subgrid, since a hierarchical partition of $\square_{abcd}$ with a horizontal or vertical cut across which there may be no districts is exactly comprised of a hierarchical partition of each subgrid. Finally, the maximum number of districts contributing to Player 1's slate from $\square_{abcd} \backslash V_-$ is the maximum number resulting from each option. Hence if we suppose that {\sc Compute-Opt}($a',b',c',d'$) has the desired behavior for all $\square_{a'b'c'd'} \subsetneq \square_{abcd}$, the claim follows by induction on the size of the subgrid $(b-a)(d-c)$ (if it reaches 1, then no recursive calls are made). 
\end{proof}

\section{Numerical Experiments} \label{numerical_experiments}

The computational problems studied in the present paper were motivated by a variety of real-world contexts in which objects are ranked. In this section, we aim to apply the results for these problems to compute the maximum and minimum rank and percentile of select objects of interest. A key point of concern is how robust the rank and percentile are to (perhaps manipulated) partitions of the ambient objects---for instance, is the maximum value close to the minimum? \\

We specialize our efforts to a salient topic: sources of greenhouse gas emissions, which is widely accepted by the scientific community as the leading cause of climate change. The relative contribution of sources to such emissions carries myriad implications for how blame is allocated, both to individual and corporate entities, and may shape policy debates on how to best mitigate climate change. The EPA itself ranks sources of greenhouse gas emissions at a coarse scale of categorization \cite{EPAranking}. \\

More specifically, we used the supplemental table ``KCA-3: 2022 Key Category Approach 1 and Approach 2 Analysis—Level Assessment, without LULUCF" in the EPA's comprehensive report \textit{Inventory of U.S. Greenhouse Gas Emissions and Sinks: 1990-2022}, which is intended to give an exhaustive breakdown of sources of emissions by CRT category, as introduced in Section \ref{hierarchical_category_variant}, with the exception of ``Land Use, Land-Use Change, and Forestry'', a CRT category that the EPA sometimes omits due to the presence of intertwined greenhouse gas emission sources and sinks (whose separation could foment a form of manipulation in rankings that we are not studying in this paper). These CRT codes allow us to model the categories according to the hierarchical category variant from Section \ref{hierarchical_category_variant}, where we also create a category for each CRT Category that is listed for multiple greenhouse gases and append the corresponding greenhouse gas to the description. Our measure $\mu$ is the EPA estimate for the amount of emissions, in million metric tons of CO$_2$ equivalent, the underlying subset of activities caused in the year 2022. \\

The results indicate that large sources of emissions are somewhat robust in their placement in rankings. For instance, emissions due to road transportation (CRT Category: 1.A.3.b Transportation: Road) are always ranked first or second in any partition, and they always appear in the top 12\% of ranked sources. More intermediate sources of emissions, however, experience much greater variability in ranking. Emissions due to cement production (CRT Category: 2.A.1 Cement Production (CO$_2$)) can be ranked as the 5th or 25th largest source of emissions and can appear after about 8\% or 71\% of other sources. Therefore, while factors such as automobiles may already dominate public discourse of personal responsibility and environmental regulations, there are a variety of other contributors to climate change whose relative importance may be manipulated by bad-faith actors and interpreted in widely different ways arising from the same set of base facts. \\

The results for additional categories are shown in Table \ref{tab:crt_category_rankings}. The partitions achieving the indicated ranks and percentiles for 2.F.1 Emissions from Substitutes for Ozone Depleting Substances: Refrigeration and Air conditioning HFCs, PFCs are reproduced in Figures \ref{fig:tree_min_rank_2f1_full_list}, \ref{fig:tree_max_rank_2f1_full_list}, and \ref{fig:tree_max_perc_2f1_full_list}. To optimize readability, node labels in these figures are factored hierarchically. Each parent node displays the maximal common string prefix of its subtree, while descendant nodes omit this inherited prefix to display only their unique differentiating suffixes. The full list of CRT categories used appears in Appendix \ref{chpt:appendix_crt_categories} along with their absolute emission contributions.

\begin{table}[H]
    \centering
    \makebox[\textwidth][c]{
        \setlength{\tabcolsep}{.7pt}
        \begin{tabular}{|>{\centering\arraybackslash}m{5.5cm}|c|c|c|c|}
            \hline
            \textbf{CRT Category} & \textbf{Min Rank} & \textbf{Max Rank} & \textbf{Min Percentile} & \textbf{Max Percentile} \\
            \hline
            1.A.3.b Transportation: Road & 1 & 2 & 0.40\% & 11.54\% \\
            \hline
            1.A.3.a Transportation: Aviation & 7 & 13 & 7.39\% & 67.86\% \\
            \hline
            2.A.1 Cement Production (CO$_2$) & 5 & 25 & 8.46\% & 71.05\% \\
            \hline
            2.F.1 Emissions from Substitutes for Ozone Depleting Substances: Refrigeration and Air conditioning HFCs, PFCs & 3 & 13 & 4.72\% & 44.74\% \\
            \hline
            3 Agriculture & 2 & 6 & 2.34\% & 50.00\% \\
            \hline
            3.B (Manure Management) & 6 & 15 & 8.87\% & 71.88\% \\
            \hline
        \end{tabular}
    }
    \caption{Minimum and maximum rank and percentile (rounded to 2 decimal places) for select EPA Common Reporting Tables (CRT) categories over all valid hierarchical category partitions.}
    \label{tab:crt_category_rankings}
\end{table}

\begin{figure}[H]
    \centering
    \adjustbox{max height=0.95\textheight, max width=\textwidth}{
    \begin{forest}
      forked edges,
      for tree={
        grow'=0, fit=band, draw, rectangle, rounded corners, align=left, 
        edge={thick}, font=\sffamily\tiny, inner sep=2pt, s sep=0.5mm, l sep=4mm
      },
      part/.style={fill=blue!10, draw=blue!80!black, thick},
      target/.style={fill=red!15, draw=red!80!black, thick}
      [{All}
        [{1 Energy}, part
          [{1.A}
            [{1.A.1 Stationary Combustion}
              [{Coal - Electricity Generation}
                [{CH$_4$}]
                [{CO$_2$}]
                [{N$_2$O}]
              ]
              [{Geothermal Energy CO$_2$}]
              [{Natural Gas - Electricity Generation}
                [{CH$_4$}]
                [{CO$_2$}]
                [{N$_2$O}]
              ]
              [{Oil - Electricity Generation}
                [{CH$_4$}]
                [{CO$_2$}]
                [{N$_2$O}]
              ]
              [{Wood - Electricity Generation}
                [{CH$_4$}]
                [{N$_2$O}]
              ]
            ]
            [{1.A.2 Stationary Combustion}
              [{Coal - Industrial CO$_2$}]
              [{Industrial}
                [{CH$_4$}]
                [{N$_2$O}]
              ]
              [{Natural Gas - Industrial CO$_2$}]
              [{Oil - Industrial CO$_2$}]
            ]
            [{1.A.3 Transportation}
              [{1.A.3.a Aviation}
                [{CH$_4$}]
                [{CO$_2$}]
                [{N$_2$O}]
              ]
              [{1.A.3.b Road}
                [{CH$_4$}]
                [{CO$_2$}]
                [{N$_2$O}]
              ]
              [{1.A.3.c Railways}
                [{CH$_4$}]
                [{CO$_2$}]
                [{N$_2$O}]
              ]
              [{1.A.3.d Domestic Navigation}
                [{CH$_4$}]
                [{CO$_2$}]
                [{N$_2$O}]
              ]
              [{1.A.3.e Other}
                [{CH$_4$}]
                [{CO$_2$}]
                [{N$_2$O}]
              ]
            ]
            [{1.A.4 Stationary Combustion}
              [{1.A.4.a}
                [{Coal - Commercial CO$_2$}]
                [{Commercial}
                  [{CH$_4$}]
                  [{N$_2$O}]
                ]
                [{Natural Gas - Commercial CO$_2$}]
                [{Oil - Commercial CO$_2$}]
              ]
              [{1.A.4.b}
                [{Coal - Residential CO$_2$}]
                [{Natural Gas - Residential CO$_2$}]
                [{Oil - Residential CO$_2$}]
                [{Residential}
                  [{CH$_4$}]
                  [{N$_2$O}]
                ]
              ]
            ]
            [{1.A.5}
              [{Non-Energy Use of Fuels CO$_2$}]
              [{Stationary Combustion - Coal - U.S. Territories CO$_2$}]
              [{Stationary Combustion - Natural Gas - U.S. Territories CO$_2$}]
              [{Stationary Combustion - Oil - U.S. Territories CO$_2$}]
              [{Stationary Combustion - U.S. Territories}
                [{CH$_4$}]
                [{N$_2$O}]
              ]
              [{1.A.5.b Transportation: Military}
                [{CH$_4$}]
                [{CO$_2$}]
                [{N$_2$O}]
              ]
            ]
          ]
          [{1.B}
            [{1.B.1}
              [{Coal Mining CO$_2$}]
              [{Fugitive Emissions from Abandoned Underground Coal Mines CH$_4$}]
              [{Fugitive Emissions from Coal Mining CH$_4$}]
            ]
            [{1.B.2}
              [{Abandoned Oil and Natural Gas Wells}
                [{CH$_4$}]
                [{CO$_2$}]
              ]
              [{Natural Gas Systems}
                [{CH$_4$}]
                [{CO$_2$}]
                [{N$_2$O}]
              ]
              [{Petroleum Systems}
                [{CH$_4$}]
                [{CO$_2$}]
                [{N$_2$O}]
              ]
            ]
          ]
        ]
        [{2 Industrial Processes and Product Use}
          [{2.A}
            [{2.A.1 Cement Production CO$_2$}, part]
            [{2.A.2 Lime Production CO$_2$}, part]
            [{2.A.3 Glass Production CO$_2$}, part]
            [{2.A.4 Other Process Uses of Carbonates CO$_2$}, part]
          ]
          [{2.B}
            [{2.B.1 Ammonia Production CO$_2$}, part]
            [{2.B.2 Nitric Acid Production N$_2$O}, part]
            [{2.B.3 Adipic Acid Production N$_2$O}, part]
            [{2.B.4 Caprolactam, Glyoxal, and Glyoxylic Acid Production N$_2$O}, part]
            [{2.B.5 Silicon Carbide Production and Consumption}
              [{CH$_4$}, part]
              [{CO$_2$}, part]
            ]
            [{2.B.6 Titanium Dioxide Production CO$_2$}, part]
            [{2.B.7 Soda Ash Production CO$_2$}, part]
            [{2.B.8 Petrochemical Production}
              [{CH$_4$}, part]
              [{CO$_2$}, part]
            ]
            [{2.B.9 Fluorochemical Production PFC, HFC, SF$_6$, NF$_3$}, part]
            [{2.B.10}
              [{Carbon Dioxide Consumption CO$_2$}]
              [{Phosphoric Acid Production CO$_2$}]
              [{Urea Consumption for Non-Ag Purposes CO$_2$}]
            ]
          ]
          [{2.C}
            [{2.C.1 Iron and Steel Production \& Metallurgical Coke Production}
              [{CH$_4$}, part]
              [{CO$_2$}, part]
            ]
            [{2.C.2 Ferroalloy Production}
              [{CH$_4$}, part]
              [{CO$_2$}, part]
            ]
            [{2.C.3 Aluminum Production}
              [{CO$_2$}, part]
              [{PFCs}, part]
            ]
            [{2.C.4 Magnesium Production and Processing}
              [{CO$_2$}, part]
              [{HFCs}, part]
              [{SF$_6$}, part]
            ]
            [{2.C.5 Lead Production CO$_2$}, part]
            [{2.C.6 Zinc Production CO$_2$}, part]
          ]
          [{2.E Electronics Industry}
            [{N$_2$O}, part]
            [{PFC, HFC, SF$_6$, NF$_3$}, part]
          ]
          [{2.F Emissions from Substitutes for Ozone Depleting Substances}
            [{2.F.1 Refrigeration and Air conditioning HFCs, PFCs}, target]
            [{2.F.2 Foam Blowing Agents HFCs, PFCs}, part]
            [{2.F.3 Fire Protection HFCs, PFCs}, part]
            [{2.F.4 Aerosols HFCs, PFCs}, part]
            [{2.F.5 Solvents HFCs, PFCs}, part]
          ]
          [{2.G}
            [{2.G Electrical Equipment PFC, SF$_6$}, part]
            [{2.G Other Product Manufacture and Use}
              [{N$_2$O}, part]
              [{PFC, HFC, SF$_6$}, part]
            ]
          ]
        ]
        [{3 Agriculture}, part
          [{3.A Enteric Fermentation}
            [{3.A.1 Cattle CH$_4$}]
            [{3.A.4 Other Livestock CH$_4$}]
          ]
          [{3.B Manure Management}
            [{3.B.1 Cattle}
              [{CH$_4$}]
              [{N$_2$O}]
            ]
            [{3.B.4 Other Livestock}
              [{CH$_4$}]
              [{N$_2$O}]
            ]
          ]
          [{3.C Rice Cultivation CH$_4$}]
          [{3.D}
            [{3.D.1 Direct Agricultural Soil Management N$_2$O}]
            [{3.D.2 Indirect Applied Nitrogen N$_2$O}]
          ]
          [{3.F Field Burning of Agricultural Residues}
            [{CH$_4$}]
            [{N$_2$O}]
          ]
          [{3.G Liming CO$_2$}]
          [{3.H Urea Fertilization CO$_2$}]
        ]
        [{5 Waste}
          [{5.A}
            [{5.A Commercial Landfills CH$_4$}, part]
            [{5.A Industrial Landfills CH$_4$}, part]
          ]
          [{5.B}
            [{5.B Composting}
              [{CH$_4$}, part]
              [{N$_2$O}, part]
            ]
            [{5.B.2 Anaerobic Digestion at Biogas Facilities CH$_4$}, part]
          ]
          [{5.C.1 Incineration of Waste}
            [{CH$_4$}, part]
            [{CO$_2$}, part]
            [{N$_2$O}, part]
          ]
          [{5.D}
            [{5.D Domestic Wastewater Treatment}
              [{CH$_4$}, part]
              [{N$_2$O}, part]
            ]
            [{5.D Industrial Wastewater Treatment}
              [{CH$_4$}, part]
              [{N$_2$O}, part]
            ]
          ]
        ]
      ]
    \end{forest}
    }
    \caption{A partition that minimizes both the rank and rank percentile of 2.F.1 Emissions from Substitutes for Ozone Depleting Substances: Refrigeration and Air conditioning HFCs, PFCs, as recorded in Table \ref{tab:crt_category_rankings}. Subsets included in the partition are shaded.}
    \label{fig:tree_min_rank_2f1_full_list}
\end{figure}
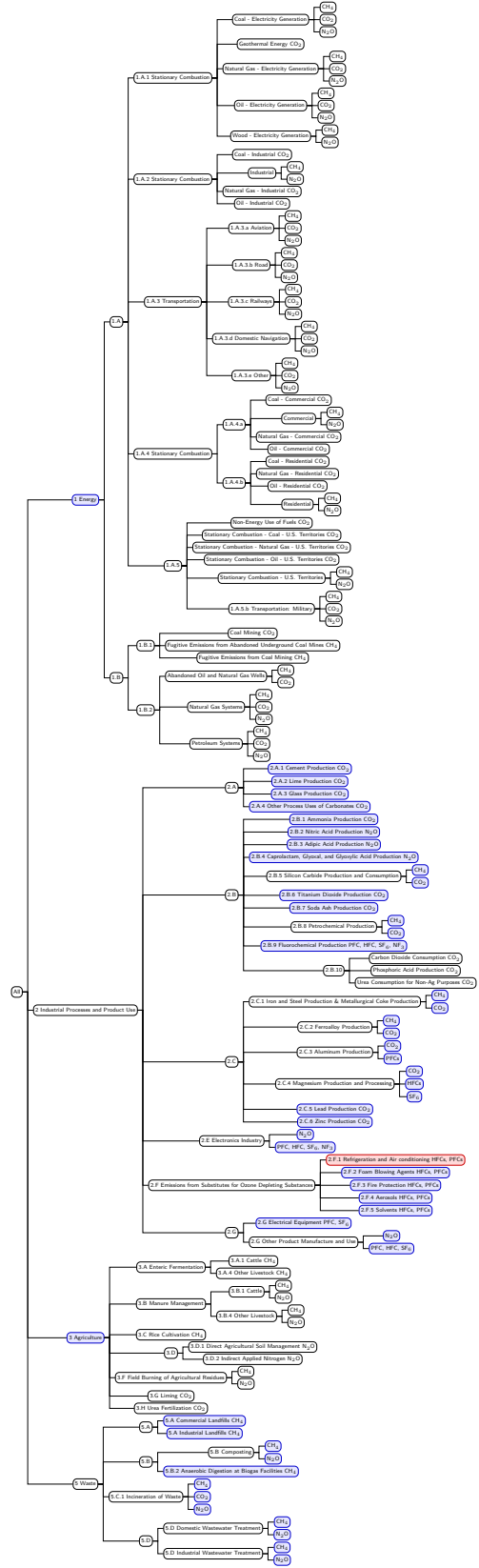

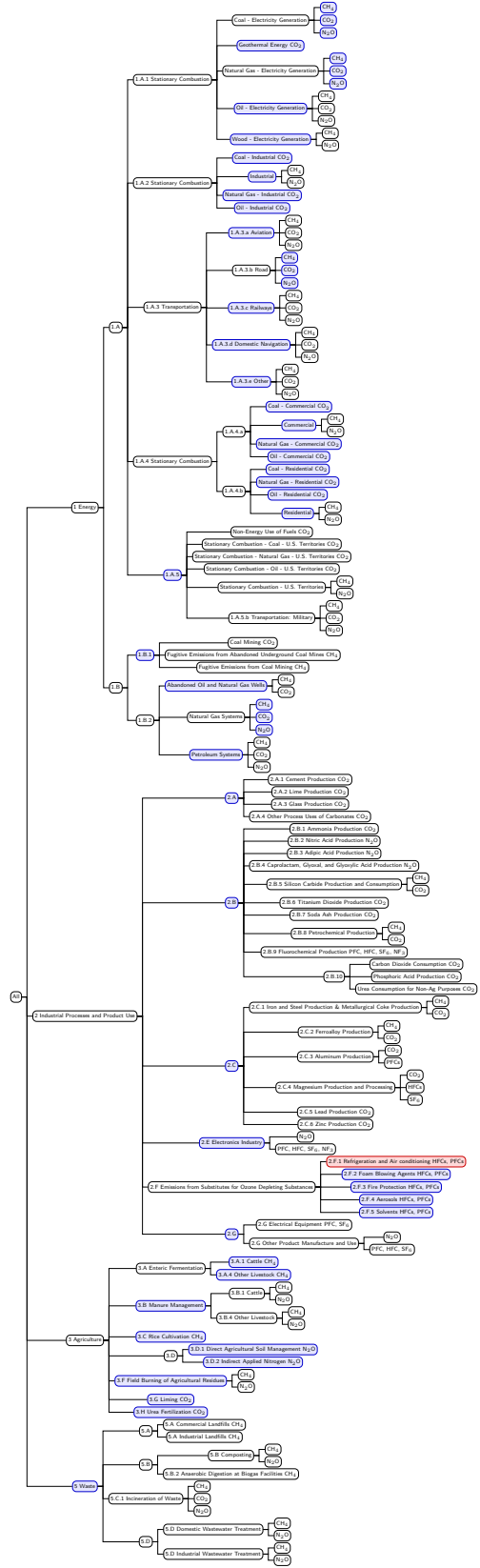
\begin{figure}[H]
    \centering
    \adjustbox{max height=0.95\textheight, max width=\textwidth}{
    \begin{forest}
      forked edges,
      for tree={
        grow'=0, fit=band, draw, rectangle, rounded corners, align=left, 
        edge={thick}, font=\sffamily\tiny, inner sep=2pt, s sep=0.5mm, l sep=4mm
      },
      part/.style={fill=blue!10, draw=blue!80!black, thick},
      target/.style={fill=red!15, draw=red!80!black, thick}
      [{All}
        [{1 Energy}
          [{1.A}
            [{1.A.1 Stationary Combustion}
              [{Coal - Electricity Generation}
                [{CH$_4$}, part]
                [{CO$_2$}, part]
                [{N$_2$O}, part]
              ]
              [{Geothermal Energy CO$_2$}, part]
              [{Natural Gas - Electricity Generation}
                [{CH$_4$}, part]
                [{CO$_2$}, part]
                [{N$_2$O}, part]
              ]
              [{Oil - Electricity Generation}, part
                [{CH$_4$}]
                [{CO$_2$}]
                [{N$_2$O}]
              ]
              [{Wood - Electricity Generation}, part
                [{CH$_4$}]
                [{N$_2$O}]
              ]
            ]
            [{1.A.2 Stationary Combustion}
              [{Coal - Industrial CO$_2$}, part]
              [{Industrial}, part
                [{CH$_4$}]
                [{N$_2$O}]
              ]
              [{Natural Gas - Industrial CO$_2$}, part]
              [{Oil - Industrial CO$_2$}, part]
            ]
            [{1.A.3 Transportation}
              [{1.A.3.a Aviation}, part
                [{CH$_4$}]
                [{CO$_2$}]
                [{N$_2$O}]
              ]
              [{1.A.3.b Road}
                [{CH$_4$}, part]
                [{CO$_2$}, part]
                [{N$_2$O}, part]
              ]
              [{1.A.3.c Railways}, part
                [{CH$_4$}]
                [{CO$_2$}]
                [{N$_2$O}]
              ]
              [{1.A.3.d Domestic Navigation}, part
                [{CH$_4$}]
                [{CO$_2$}]
                [{N$_2$O}]
              ]
              [{1.A.3.e Other}, part
                [{CH$_4$}]
                [{CO$_2$}]
                [{N$_2$O}]
              ]
            ]
            [{1.A.4 Stationary Combustion}
              [{1.A.4.a}
                [{Coal - Commercial CO$_2$}, part]
                [{Commercial}, part
                  [{CH$_4$}]
                  [{N$_2$O}]
                ]
                [{Natural Gas - Commercial CO$_2$}, part]
                [{Oil - Commercial CO$_2$}, part]
              ]
              [{1.A.4.b}
                [{Coal - Residential CO$_2$}, part]
                [{Natural Gas - Residential CO$_2$}, part]
                [{Oil - Residential CO$_2$}, part]
                [{Residential}, part
                  [{CH$_4$}]
                  [{N$_2$O}]
                ]
              ]
            ]
            [{1.A.5}, part
              [{Non-Energy Use of Fuels CO$_2$}]
              [{Stationary Combustion - Coal - U.S. Territories CO$_2$}]
              [{Stationary Combustion - Natural Gas - U.S. Territories CO$_2$}]
              [{Stationary Combustion - Oil - U.S. Territories CO$_2$}]
              [{Stationary Combustion - U.S. Territories}
                [{CH$_4$}]
                [{N$_2$O}]
              ]
              [{1.A.5.b Transportation: Military}
                [{CH$_4$}]
                [{CO$_2$}]
                [{N$_2$O}]
              ]
            ]
          ]
          [{1.B}
            [{1.B.1}, part
              [{Coal Mining CO$_2$}]
              [{Fugitive Emissions from Abandoned Underground Coal Mines CH$_4$}]
              [{Fugitive Emissions from Coal Mining CH$_4$}]
            ]
            [{1.B.2}
              [{Abandoned Oil and Natural Gas Wells}, part
                [{CH$_4$}]
                [{CO$_2$}]
              ]
              [{Natural Gas Systems}
                [{CH$_4$}, part]
                [{CO$_2$}, part]
                [{N$_2$O}, part]
              ]
              [{Petroleum Systems}, part
                [{CH$_4$}]
                [{CO$_2$}]
                [{N$_2$O}]
              ]
            ]
          ]
        ]
        [{2 Industrial Processes and Product Use}
          [{2.A}, part
            [{2.A.1 Cement Production CO$_2$}]
            [{2.A.2 Lime Production CO$_2$}]
            [{2.A.3 Glass Production CO$_2$}]
            [{2.A.4 Other Process Uses of Carbonates CO$_2$}]
          ]
          [{2.B}, part
            [{2.B.1 Ammonia Production CO$_2$}]
            [{2.B.2 Nitric Acid Production N$_2$O}]
            [{2.B.3 Adipic Acid Production N$_2$O}]
            [{2.B.4 Caprolactam, Glyoxal, and Glyoxylic Acid Production N$_2$O}]
            [{2.B.5 Silicon Carbide Production and Consumption}
              [{CH$_4$}]
              [{CO$_2$}]
            ]
            [{2.B.6 Titanium Dioxide Production CO$_2$}]
            [{2.B.7 Soda Ash Production CO$_2$}]
            [{2.B.8 Petrochemical Production}
              [{CH$_4$}]
              [{CO$_2$}]
            ]
            [{2.B.9 Fluorochemical Production PFC, HFC, SF$_6$, NF$_3$}]
            [{2.B.10}
              [{Carbon Dioxide Consumption CO$_2$}]
              [{Phosphoric Acid Production CO$_2$}]
              [{Urea Consumption for Non-Ag Purposes CO$_2$}]
            ]
          ]
          [{2.C}, part
            [{2.C.1 Iron and Steel Production \& Metallurgical Coke Production}
              [{CH$_4$}]
              [{CO$_2$}]
            ]
            [{2.C.2 Ferroalloy Production}
              [{CH$_4$}]
              [{CO$_2$}]
            ]
            [{2.C.3 Aluminum Production}
              [{CO$_2$}]
              [{PFCs}]
            ]
            [{2.C.4 Magnesium Production and Processing}
              [{CO$_2$}]
              [{HFCs}]
              [{SF$_6$}]
            ]
            [{2.C.5 Lead Production CO$_2$}]
            [{2.C.6 Zinc Production CO$_2$}]
          ]
          [{2.E Electronics Industry}, part
            [{N$_2$O}]
            [{PFC, HFC, SF$_6$, NF$_3$}]
          ]
          [{2.F Emissions from Substitutes for Ozone Depleting Substances}
            [{2.F.1 Refrigeration and Air conditioning HFCs, PFCs}, target]
            [{2.F.2 Foam Blowing Agents HFCs, PFCs}, part]
            [{2.F.3 Fire Protection HFCs, PFCs}, part]
            [{2.F.4 Aerosols HFCs, PFCs}, part]
            [{2.F.5 Solvents HFCs, PFCs}, part]
          ]
          [{2.G}, part
            [{2.G Electrical Equipment PFC, SF$_6$}]
            [{2.G Other Product Manufacture and Use}
              [{N$_2$O}]
              [{PFC, HFC, SF$_6$}]
            ]
          ]
        ]
        [{3 Agriculture}
          [{3.A Enteric Fermentation}
            [{3.A.1 Cattle CH$_4$}, part]
            [{3.A.4 Other Livestock CH$_4$}, part]
          ]
          [{3.B Manure Management}, part
            [{3.B.1 Cattle}
              [{CH$_4$}]
              [{N$_2$O}]
            ]
            [{3.B.4 Other Livestock}
              [{CH$_4$}]
              [{N$_2$O}]
            ]
          ]
          [{3.C Rice Cultivation CH$_4$}, part]
          [{3.D}
            [{3.D.1 Direct Agricultural Soil Management N$_2$O}, part]
            [{3.D.2 Indirect Applied Nitrogen N$_2$O}, part]
          ]
          [{3.F Field Burning of Agricultural Residues}, part
            [{CH$_4$}]
            [{N$_2$O}]
          ]
          [{3.G Liming CO$_2$}, part]
          [{3.H Urea Fertilization CO$_2$}, part]
        ]
        [{5 Waste}, part
          [{5.A}
            [{5.A Commercial Landfills CH$_4$}]
            [{5.A Industrial Landfills CH$_4$}]
          ]
          [{5.B}
            [{5.B Composting}
              [{CH$_4$}]
              [{N$_2$O}]
            ]
            [{5.B.2 Anaerobic Digestion at Biogas Facilities CH$_4$}]
          ]
          [{5.C.1 Incineration of Waste}
            [{CH$_4$}]
            [{CO$_2$}]
            [{N$_2$O}]
          ]
          [{5.D}
            [{5.D Domestic Wastewater Treatment}
              [{CH$_4$}]
              [{N$_2$O}]
            ]
            [{5.D Industrial Wastewater Treatment}
              [{CH$_4$}]
              [{N$_2$O}]
            ]
          ]
        ]
      ]
    \end{forest}
    }
    \caption{A partition that maximizes the rank of 2.F.1 Emissions from Substitutes for Ozone Depleting Substances: Refrigeration and Air conditioning HFCs, PFCs, as recorded in Table \ref{tab:crt_category_rankings}. Subsets included in the partition are shaded.}
    \label{fig:tree_max_rank_2f1_full_list}
\end{figure}


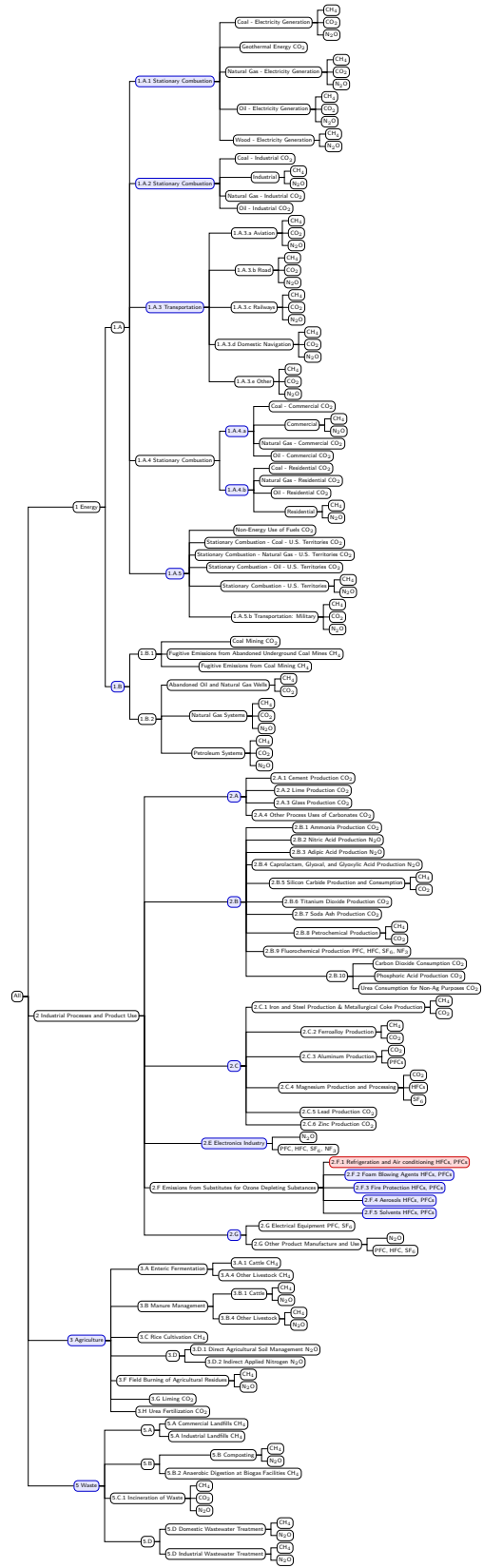
\begin{figure}[H] 
    \centering
    \adjustbox{max height=0.95\textheight, max width=\textwidth}{
    \begin{forest}
      forked edges,
      for tree={
        grow'=0, fit=band, draw, rectangle, rounded corners, align=left, 
        edge={thick}, font=\sffamily\tiny, inner sep=2pt, s sep=0.5mm, l sep=4mm
      },
      part/.style={fill=blue!10, draw=blue!80!black, thick},
      target/.style={fill=red!15, draw=red!80!black, thick}
      [{All}
        [{1 Energy}
          [{1.A}
            [{1.A.1 Stationary Combustion}, part
              [{Coal - Electricity Generation}
                [{CH$_4$}]
                [{CO$_2$}]
                [{N$_2$O}]
              ]
              [{Geothermal Energy CO$_2$}]
              [{Natural Gas - Electricity Generation}
                [{CH$_4$}]
                [{CO$_2$}]
                [{N$_2$O}]
              ]
              [{Oil - Electricity Generation}
                [{CH$_4$}]
                [{CO$_2$}]
                [{N$_2$O}]
              ]
              [{Wood - Electricity Generation}
                [{CH$_4$}]
                [{N$_2$O}]
              ]
            ]
            [{1.A.2 Stationary Combustion}, part
              [{Coal - Industrial CO$_2$}]
              [{Industrial}
                [{CH$_4$}]
                [{N$_2$O}]
              ]
              [{Natural Gas - Industrial CO$_2$}]
              [{Oil - Industrial CO$_2$}]
            ]
            [{1.A.3 Transportation}, part
              [{1.A.3.a Aviation}
                [{CH$_4$}]
                [{CO$_2$}]
                [{N$_2$O}]
              ]
              [{1.A.3.b Road}
                [{CH$_4$}]
                [{CO$_2$}]
                [{N$_2$O}]
              ]
              [{1.A.3.c Railways}
                [{CH$_4$}]
                [{CO$_2$}]
                [{N$_2$O}]
              ]
              [{1.A.3.d Domestic Navigation}
                [{CH$_4$}]
                [{CO$_2$}]
                [{N$_2$O}]
              ]
              [{1.A.3.e Other}
                [{CH$_4$}]
                [{CO$_2$}]
                [{N$_2$O}]
              ]
            ]
            [{1.A.4 Stationary Combustion}
              [{1.A.4.a}, part
                [{Coal - Commercial CO$_2$}]
                [{Commercial}
                  [{CH$_4$}]
                  [{N$_2$O}]
                ]
                [{Natural Gas - Commercial CO$_2$}]
                [{Oil - Commercial CO$_2$}]
              ]
              [{1.A.4.b}, part
                [{Coal - Residential CO$_2$}]
                [{Natural Gas - Residential CO$_2$}]
                [{Oil - Residential CO$_2$}]
                [{Residential}
                  [{CH$_4$}]
                  [{N$_2$O}]
                ]
              ]
            ]
            [{1.A.5}, part
              [{Non-Energy Use of Fuels CO$_2$}]
              [{Stationary Combustion - Coal - U.S. Territories CO$_2$}]
              [{Stationary Combustion - Natural Gas - U.S. Territories CO$_2$}]
              [{Stationary Combustion - Oil - U.S. Territories CO$_2$}]
              [{Stationary Combustion - U.S. Territories}
                [{CH$_4$}]
                [{N$_2$O}]
              ]
              [{1.A.5.b Transportation: Military}
                [{CH$_4$}]
                [{CO$_2$}]
                [{N$_2$O}]
              ]
            ]
          ]
          [{1.B}, part
            [{1.B.1}
              [{Coal Mining CO$_2$}]
              [{Fugitive Emissions from Abandoned Underground Coal Mines CH$_4$}]
              [{Fugitive Emissions from Coal Mining CH$_4$}]
            ]
            [{1.B.2}
              [{Abandoned Oil and Natural Gas Wells}
                [{CH$_4$}]
                [{CO$_2$}]
              ]
              [{Natural Gas Systems}
                [{CH$_4$}]
                [{CO$_2$}]
                [{N$_2$O}]
              ]
              [{Petroleum Systems}
                [{CH$_4$}]
                [{CO$_2$}]
                [{N$_2$O}]
              ]
            ]
          ]
        ]
        [{2 Industrial Processes and Product Use}
          [{2.A}, part
            [{2.A.1 Cement Production CO$_2$}]
            [{2.A.2 Lime Production CO$_2$}]
            [{2.A.3 Glass Production CO$_2$}]
            [{2.A.4 Other Process Uses of Carbonates CO$_2$}]
          ]
          [{2.B}, part
            [{2.B.1 Ammonia Production CO$_2$}]
            [{2.B.2 Nitric Acid Production N$_2$O}]
            [{2.B.3 Adipic Acid Production N$_2$O}]
            [{2.B.4 Caprolactam, Glyoxal, and Glyoxylic Acid Production N$_2$O}]
            [{2.B.5 Silicon Carbide Production and Consumption}
              [{CH$_4$}]
              [{CO$_2$}]
            ]
            [{2.B.6 Titanium Dioxide Production CO$_2$}]
            [{2.B.7 Soda Ash Production CO$_2$}]
            [{2.B.8 Petrochemical Production}
              [{CH$_4$}]
              [{CO$_2$}]
            ]
            [{2.B.9 Fluorochemical Production PFC, HFC, SF$_6$, NF$_3$}]
            [{2.B.10}
              [{Carbon Dioxide Consumption CO$_2$}]
              [{Phosphoric Acid Production CO$_2$}]
              [{Urea Consumption for Non-Ag Purposes CO$_2$}]
            ]
          ]
          [{2.C}, part
            [{2.C.1 Iron and Steel Production \& Metallurgical Coke Production}
              [{CH$_4$}]
              [{CO$_2$}]
            ]
            [{2.C.2 Ferroalloy Production}
              [{CH$_4$}]
              [{CO$_2$}]
            ]
            [{2.C.3 Aluminum Production}
              [{CO$_2$}]
              [{PFCs}]
            ]
            [{2.C.4 Magnesium Production and Processing}
              [{CO$_2$}]
              [{HFCs}]
              [{SF$_6$}]
            ]
            [{2.C.5 Lead Production CO$_2$}]
            [{2.C.6 Zinc Production CO$_2$}]
          ]
          [{2.E Electronics Industry}, part
            [{N$_2$O}]
            [{PFC, HFC, SF$_6$, NF$_3$}]
          ]
          [{2.F Emissions from Substitutes for Ozone Depleting Substances}
            [{2.F.1 Refrigeration and Air conditioning HFCs, PFCs}, target]
            [{2.F.2 Foam Blowing Agents HFCs, PFCs}, part]
            [{2.F.3 Fire Protection HFCs, PFCs}, part]
            [{2.F.4 Aerosols HFCs, PFCs}, part]
            [{2.F.5 Solvents HFCs, PFCs}, part]
          ]
          [{2.G}, part
            [{2.G Electrical Equipment PFC, SF$_6$}]
            [{2.G Other Product Manufacture and Use}
              [{N$_2$O}]
              [{PFC, HFC, SF$_6$}]
            ]
          ]
        ]
        [{3 Agriculture}, part
          [{3.A Enteric Fermentation}
            [{3.A.1 Cattle CH$_4$}]
            [{3.A.4 Other Livestock CH$_4$}]
          ]
          [{3.B Manure Management}
            [{3.B.1 Cattle}
              [{CH$_4$}]
              [{N$_2$O}]
            ]
            [{3.B.4 Other Livestock}
              [{CH$_4$}]
              [{N$_2$O}]
            ]
          ]
          [{3.C Rice Cultivation CH$_4$}]
          [{3.D}
            [{3.D.1 Direct Agricultural Soil Management N$_2$O}]
            [{3.D.2 Indirect Applied Nitrogen N$_2$O}]
          ]
          [{3.F Field Burning of Agricultural Residues}
            [{CH$_4$}]
            [{N$_2$O}]
          ]
          [{3.G Liming CO$_2$}]
          [{3.H Urea Fertilization CO$_2$}]
        ]
        [{5 Waste}, part
          [{5.A}
            [{5.A Commercial Landfills CH$_4$}]
            [{5.A Industrial Landfills CH$_4$}]
          ]
          [{5.B}
            [{5.B Composting}
              [{CH$_4$}]
              [{N$_2$O}]
            ]
            [{5.B.2 Anaerobic Digestion at Biogas Facilities CH$_4$}]
          ]
          [{5.C.1 Incineration of Waste}
            [{CH$_4$}]
            [{CO$_2$}]
            [{N$_2$O}]
          ]
          [{5.D}
            [{5.D Domestic Wastewater Treatment}
              [{CH$_4$}]
              [{N$_2$O}]
            ]
            [{5.D Industrial Wastewater Treatment}
              [{CH$_4$}]
              [{N$_2$O}]
            ]
          ]
        ]
      ]
    \end{forest}
    }
    \caption{A partition that maximizes the rank percentile of 2.F.1 Emissions from Substitutes for Ozone Depleting Substances: Refrigeration and Air conditioning HFCs, PFCs, as recorded in Table \ref{tab:crt_category_rankings}. Subsets included in the partition are shaded.}
    \label{fig:tree_max_perc_2f1_full_list}
\end{figure}






\newpage

\section{Future Work} \label{future_work}
The above treatment leaves several opportunities to strengthen the results presented therein. \\

For instance, the approximation guarantees may be able to be improved with the additional settings of the above special cases. The $\frac{k+2}{3}$-approximation algorithm for $k$-{\sc Set Packing} holds for an arbitrary set-system with $n$ subsets; it uses a local search algorithm with swaps of size $O(\log n)$. Assuming the structure of a graph---or in particular a $(k-1)$-regular graph where a vertex and its neighbors always form a $k$-component---may affect the guarantees possible with local search. It is worth noting that the authors responsible for the above guarantee do show that it is essentially sharp for local search algorithms that consider swap sizes that are linear in $n$, but the example is not an instance of the uniform value case of {\sc Rank Maximization} \cite{localsearchapprox}. Preliminary results suggest that, in a $(k-1)$-regular graph, local search may be used to obtain a constant fraction of the total possible $\lfloor \frac{n}{k} \rfloor$ disjoint $k$-components that may be in a graph. It also appears that using the probabilistic method to randomly place vertices in subgraphs of size $ck$ for some $c \in \{2,3,\dots\}$ may be used to show the existence of a constant fraction of $\lfloor \frac{n}{k} \rfloor$ disjoint $k$-components in any $(k-1)$-regular graph. Nonetheless, some additional care is required to handle the conditional probabilities involved. \\

Similarly, exploring a special case of the partition problems where the input graph is planar would enable us to investigate perennial ordinal statements involving geography, e.g., that California has the 5th largest economy in the world \cite{govca2024economy}. This claim is based on treating California as a country (along with the remainder of the United States) yet otherwise maintaining the list of recognized countries by the International Monetary Fund, which already involves counterfactual reasoning that may be taken to its limit via the above framework. \\

Further bounds on the maximum and minimum rank may be pursued in more of the special cases and variants presented above. This is consistent with the overarching goal of the project: to study how robust rank and percentile are to partitioning by (potentially bad-faith) actors. Of particular interest are frequently encountered claims about societal problems made by news media or corporations. This includes analysis of social data in the United States by statistical area, Netflix's data of most-watched series by season, and nutritional labels for processed foods. 

\section*{Acknowledgments}
The author would like to thank Martin Strauss, Mahdi Cheraghchi, and Euiwoong Lee for many helpful conversations and insights. 

\printbibliography

\newpage

\appendix

\section{List of CRT Categories} \label{chpt:appendix_crt_categories}

The following list gives the hierarchical breakdown of CRT Categories, with 2022-level emissions, in million metric tons of $\mathrm{CO_2}$ equivalent, shown in parentheses. Indentation indicates a subcategory of the category immediately to the left. \\

\treeitem{0}{All (6343.21)}
\treeitem{1}{1 Energy (5187.08)}
\treeitem{2}{1.A (4854.82)}
\treeitem{3}{1.A.1 (1554.53)}
\treeitem{4}{1.A.1 Stationary Combustion - Coal - Electricity Generation (869.853)}
\treeitem{5}{1.A.1 Stationary Combustion - Coal - Electricity Generation CH$_4$ (0.228783)}
\treeitem{5}{1.A.1 Stationary Combustion - Coal - Electricity Generation CO$_2$ (851.472)}
\treeitem{5}{1.A.1 Stationary Combustion - Coal - Electricity Generation N$_2$O (18.1526)}
\treeitem{4}{1.A.1 Stationary Combustion - Geothermal Energy CO$_2$ (0.380557)}
\treeitem{4}{1.A.1 Stationary Combustion - Natural Gas - Electricity Generation (663.75)}
\treeitem{5}{1.A.1 Stationary Combustion - Natural Gas - Electricity Generation CH$_4$ (1.03868)}
\treeitem{5}{1.A.1 Stationary Combustion - Natural Gas - Electricity Generation CO$_2$ (659.305)}
\treeitem{5}{1.A.1 Stationary Combustion - Natural Gas - Electricity Generation N$_2$O (3.40572)}
\treeitem{4}{1.A.1 Stationary Combustion - Oil - Electricity Generation (20.5295)}
\treeitem{5}{1.A.1 Stationary Combustion - Oil - Electricity Generation CH$_4$ (0.00149642)}
\treeitem{5}{1.A.1 Stationary Combustion - Oil - Electricity Generation CO$_2$ (20.5222)}
\treeitem{5}{1.A.1 Stationary Combustion - Oil - Electricity Generation N$_2$O (0.00580066)}
\treeitem{4}{1.A.1 Stationary Combustion - Wood - Electricity Generation (0.0191025)}
\treeitem{5}{1.A.1 Stationary Combustion - Wood - Electricity Generation CH$_4$ (0.0018255)}
\treeitem{5}{1.A.1 Stationary Combustion - Wood - Electricity Generation N$_2$O (0.017277)}
\treeitem{3}{1.A.2 (804.678)}
\treeitem{4}{1.A.2 Stationary Combustion - Coal - Industrial CO$_2$ (43.0351)}
\treeitem{4}{1.A.2 Stationary Combustion - Industrial (3.61476)}
\treeitem{5}{1.A.2 Stationary Combustion - Industrial CH$_4$ (1.58172)}
\treeitem{5}{1.A.2 Stationary Combustion - Industrial N$_2$O (2.03305)}
\treeitem{4}{1.A.2 Stationary Combustion - Natural Gas - Industrial CO$_2$ (510.385)}
\treeitem{4}{1.A.2 Stationary Combustion - Oil - Industrial CO$_2$ (247.644)}
\treeitem{3}{1.A.3 (1765.82)}
\treeitem{4}{1.A.3.a Transportation: Aviation (167.002)}
\treeitem{5}{1.A.3.a Transportation: Aviation CH$_4$ (0.0367006)}
\treeitem{5}{1.A.3.a Transportation: Aviation CO$_2$ (165.612)}
\treeitem{5}{1.A.3.a Transportation: Aviation N$_2$O (1.35306)}
\treeitem{4}{1.A.3.b Transportation: Road (1447.94)}
\treeitem{5}{1.A.3.b Transportation: Road CH$_4$ (0.910629)}
\treeitem{5}{1.A.3.b Transportation: Road CO$_2$ (1438.14)}
\treeitem{5}{1.A.3.b Transportation: Road N$_2$O (8.89341)}
\treeitem{4}{1.A.3.c Transportation: Railways (32.826)}
\treeitem{5}{1.A.3.c Transportation: Railways CH$_4$ (0.0752349)}
\treeitem{5}{1.A.3.c Transportation: Railways CO$_2$ (32.5229)}
\treeitem{5}{1.A.3.c Transportation: Railways N$_2$O (0.227854)}
\treeitem{4}{1.A.3.d Transportation: Domestic Navigation (41.6155)}
\treeitem{5}{1.A.3.d Transportation: Domestic Navigation CH$_4$ (0.469529)}
\treeitem{5}{1.A.3.d Transportation: Domestic Navigation CO$_2$ (40.8857)}
\treeitem{5}{1.A.3.d Transportation: Domestic Navigation N$_2$O (0.260323)}
\treeitem{4}{1.A.3.e Transportation: Other (76.4292)}
\treeitem{5}{1.A.3.e Transportation: Other CH$_4$ (1.11726)}
\treeitem{5}{1.A.3.e Transportation: Other CO$_2$ (69.3499)}
\treeitem{5}{1.A.3.e Transportation: Other N$_2$O (5.96204)}
\treeitem{3}{1.A.4 (599.547)}
\treeitem{4}{1.A.4.a (260.473)}
\treeitem{5}{1.A.4.a Stationary Combustion - Coal - Commercial CO$_2$ (1.38813)}
\treeitem{5}{1.A.4.a Stationary Combustion - Commercial (1.73966)}
\treeitem{6}{1.A.4.a Stationary Combustion - Commercial CH$_4$ (1.41059)}
\treeitem{6}{1.A.4.a Stationary Combustion - Commercial N$_2$O (0.329069)}
\treeitem{5}{1.A.4.a Stationary Combustion - Natural Gas - Commercial CO$_2$ (192.262)}
\treeitem{5}{1.A.4.a Stationary Combustion - Oil - Commercial CO$_2$ (65.0824)}
\treeitem{4}{1.A.4.b (339.074)}
\treeitem{5}{1.A.4.b Stationary Combustion - Coal - Residential CO$_2$ (0)}
\treeitem{5}{1.A.4.b Stationary Combustion - Natural Gas - Residential CO$_2$ (271.987)}
\treeitem{5}{1.A.4.b Stationary Combustion - Oil - Residential CO$_2$ (62.0785)}
\treeitem{5}{1.A.4.b Stationary Combustion - Residential (5.00875)}
\treeitem{6}{1.A.4.b Stationary Combustion - Residential CH$_4$ (4.30891)}
\treeitem{6}{1.A.4.b Stationary Combustion - Residential N$_2$O (0.699838)}
\treeitem{3}{1.A.5 (130.245)}
\treeitem{4}{1.A.5 Non-Energy Use of Fuels CO$_2$ (102.808)}
\treeitem{4}{1.A.5 Stationary Combustion - Coal - U.S. Territories CO$_2$ (2.88857)}
\treeitem{4}{1.A.5 Stationary Combustion - Natural Gas - U.S. Territories CO$_2$ (2.72723)}
\treeitem{4}{1.A.5 Stationary Combustion - Oil - U.S. Territories CO$_2$ (16.9588)}
\treeitem{4}{1.A.5 Stationary Combustion - U.S. Territories (0.0861103)}
\treeitem{5}{1.A.5 Stationary Combustion - U.S. Territories CH$_4$ (0.0350626)}
\treeitem{5}{1.A.5 Stationary Combustion - U.S. Territories N$_2$O (0.0510477)}
\treeitem{4}{1.A.5.b Transportation: Military (4.77658)}
\treeitem{5}{1.A.5.b Transportation: Military CH$_4$ (0.000279364)}
\treeitem{5}{1.A.5.b Transportation: Military CO$_2$ (4.77618)}
\treeitem{5}{1.A.5.b Transportation: Military N$_2$O (0.00012299)}
\treeitem{2}{1.B (332.264)}
\treeitem{3}{1.B.1 (52.3999)}
\treeitem{4}{1.B.1 Coal Mining CO$_2$ (2.47399)}
\treeitem{4}{1.B.1 Fugitive Emissions from Abandoned Underground Coal Mines CH$_4$ (6.29926)}
\treeitem{4}{1.B.1 Fugitive Emissions from Coal Mining CH$_4$ (43.6266)}
\treeitem{3}{1.B.2 (279.864)}
\treeitem{4}{1.B.2 Abandoned Oil and Natural Gas Wells (8.50287)}
\treeitem{5}{1.B.2 Abandoned Oil and Natural Gas Wells CH$_4$ (8.49511)}
\treeitem{5}{1.B.2 Abandoned Oil and Natural Gas Wells CO$_2$ (0.00775717)}
\treeitem{4}{1.B.2 Natural Gas Systems (209.733)}
\treeitem{5}{1.B.2 Natural Gas Systems CH$_4$ (173.111)}
\treeitem{5}{1.B.2 Natural Gas Systems CO$_2$ (36.47)}
\treeitem{5}{1.B.2 Natural Gas Systems N$_2$O (0.151986)}
\treeitem{4}{1.B.2 Petroleum Systems (61.6286)}
\treeitem{5}{1.B.2 Petroleum Systems CH$_4$ (39.6145)}
\treeitem{5}{1.B.2 Petroleum Systems CO$_2$ (21.9666)}
\treeitem{5}{1.B.2 Petroleum Systems N$_2$O (0.0475097)}
\treeitem{1}{2 Industrial Processes and Product Use (383.183)}
\treeitem{2}{2.A (66.4314)}
\treeitem{3}{2.A.1 Cement Production CO$_2$ (41.8844)}
\treeitem{3}{2.A.2 Lime Production CO$_2$ (12.2075)}
\treeitem{3}{2.A.3 Glass Production CO$_2$ (1.9558)}
\treeitem{3}{2.A.4 Other Process Uses of Carbonates CO$_2$ (10.3836)}
\treeitem{2}{2.B (77.5157)}
\treeitem{3}{2.B.1 Ammonia Production CO$_2$ (12.6098)}
\treeitem{3}{2.B.2 Nitric Acid Production N$_2$O (8.6125)}
\treeitem{3}{2.B.3 Adipic Acid Production N$_2$O (2.0888)}
\treeitem{3}{2.B.4 Caprolactam, Glyoxal, and Glyoxylic Acid Production N$_2$O (1.3356)}
\treeitem{3}{2.B.5 Silicon Carbide Production and Consumption (0.222757)}
\treeitem{4}{2.B.5 Silicon Carbide Production and Consumption CH$_4$ (0.012992)}
\treeitem{4}{2.B.5 Silicon Carbide Production and Consumption CO$_2$ (0.209765)}
\treeitem{3}{2.B.6 Titanium Dioxide Production CO$_2$ (1.474)}
\treeitem{3}{2.B.7 Soda Ash Production CO$_2$ (1.70399)}
\treeitem{3}{2.B.8 Petrochemical Production (28.7928)}
\treeitem{4}{2.B.8 Petrochemical Production CH$_4$ (0.004788)}
\treeitem{4}{2.B.8 Petrochemical Production CO$_2$ (28.788)}
\treeitem{3}{2.B.9 Fluorochemical Production PFC, HFC, SF$_6$, NF$_3$ (7.7828)}
\treeitem{3}{2.B.10 (12.8927)}
\treeitem{4}{2.B.10 Carbon Dioxide Consumption CO$_2$ (5)}
\treeitem{4}{2.B.10 Phosphoric Acid Production CO$_2$ (0.84009)}
\treeitem{4}{2.B.10 Urea Consumption for Non-Ag Purposes CO$_2$ (7.05256)}
\treeitem{2}{2.C (46.7469)}
\treeitem{3}{2.C.1 Iron and Steel Production \& Metallurgical Coke Production (40.6796)}
\treeitem{4}{2.C.1 Iron and Steel Production \& Metallurgical Coke Production CH$_4$ (0.00771267)}
\treeitem{4}{2.C.1 Iron and Steel Production \& Metallurgical Coke Production CO$_2$ (40.6719)}
\treeitem{3}{2.C.2 Ferroalloy Production (1.33736)}
\treeitem{4}{2.C.2 Ferroalloy Production CH$_4$ (0.0104055)}
\treeitem{4}{2.C.2 Ferroalloy Production CO$_2$ (1.32695)}
\treeitem{3}{2.C.3 Aluminum Production (2.20261)}
\treeitem{4}{2.C.3 Aluminum Production CO$_2$ (1.44634)}
\treeitem{4}{2.C.3 Aluminum Production PFCs (0.756268)}
\treeitem{3}{2.C.4 Magnesium Production and Processing (1.15229)}
\treeitem{4}{2.C.4 Magnesium Production and Processing CO$_2$ (0.00294143)}
\treeitem{4}{2.C.4 Magnesium Production and Processing HFCs (0.0288021)}
\treeitem{4}{2.C.4 Magnesium Production and Processing SF$_6$ (1.12054)}
\treeitem{3}{2.C.5 Lead Production CO$_2$ (0.4275)}
\treeitem{3}{2.C.6 Zinc Production CO$_2$ (0.947466)}
\treeitem{2}{2.E Electronics Industry (4.73432)}
\treeitem{3}{2.E Electronics Industry N$_2$O (0.295248)}
\treeitem{3}{2.E Electronics Industry PFC, HFC, SF$_6$, NF$_3$ (4.43907)}
\treeitem{2}{2.F (178.134)}
\treeitem{3}{2.F.1 Emissions from Substitutes for Ozone Depleting Substances: Refrigeration and Air conditioning HFCs, PFCs (144.637)}
\treeitem{3}{2.F.2 Emissions from Substitutes for Ozone Depleting Substances: Foam Blowing Agents HFCs, PFCs (11.6871)}
\treeitem{3}{2.F.3 Emissions from Substitutes for Ozone Depleting Substances: Fire Protection HFCs, PFCs (2.64203)}
\treeitem{3}{2.F.4 Emissions from Substitutes for Ozone Depleting Substances: Aerosols HFCs, PFCs (17.0371)}
\treeitem{3}{2.F.5 Emissions from Substitutes for Ozone Depleting Substances: Solvents HFCs, PFCs (2.13094)}
\treeitem{2}{2.G (9.62075)}
\treeitem{3}{2.G Electrical Equipment PFC, SF$_6$ (5.07839)}
\treeitem{3}{2.G Other Product Manufacture and Use (4.54236)}
\treeitem{4}{2.G Other Product Manufacture and Use N$_2$O (3.7503)}
\treeitem{4}{2.G Other Product Manufacture and Use PFC, HFC, SF$_6$ (0.792061)}
\treeitem{1}{3 Agriculture (593.383)}
\treeitem{2}{3.A (192.578)}
\treeitem{3}{3.A.1 Enteric Fermentation: Cattle CH$_4$ (185.9)}
\treeitem{3}{3.A.4 Enteric Fermentation: Other Livestock CH$_4$ (6.67784)}
\treeitem{2}{3.B (81.7182)}
\treeitem{3}{3.B.1 Manure Management: Cattle (50.3113)}
\treeitem{4}{3.B.1 Manure Management: Cattle CH$_4$ (37.7265)}
\treeitem{4}{3.B.1 Manure Management: Cattle N$_2$O (12.5848)}
\treeitem{3}{3.B.4 Manure Management: Other Livestock (31.4069)}
\treeitem{4}{3.B.4 Manure Management: Other Livestock CH$_4$ (26.998)}
\treeitem{4}{3.B.4 Manure Management: Other Livestock N$_2$O (4.40883)}
\treeitem{2}{3.C Rice Cultivation CH$_4$ (18.8673)}
\treeitem{2}{3.D (290.802)}
\treeitem{3}{3.D.1 Direct Agricultural Soil Management N$_2$O (262.477)}
\treeitem{3}{3.D.2 Indirect Applied Nitrogen N$_2$O (28.3245)}
\treeitem{2}{3.F Field Burning of Agricultural Residues (0.822987)}
\treeitem{3}{3.F Field Burning of Agricultural Residues CH$_4$ (0.620643)}
\treeitem{3}{3.F Field Burning of Agricultural Residues N$_2$O (0.202344)}
\treeitem{2}{3.G Liming CO$_2$ (3.268)}
\treeitem{2}{3.H Urea Fertilization CO$_2$ (5.32742)}
\treeitem{1}{5 Waste (179.56)}
\treeitem{2}{5.A (119.767)}
\treeitem{3}{5.A Commercial Landfills CH$_4$ (100.856)}
\treeitem{3}{5.A Industrial Landfills CH$_4$ (18.9112)}
\treeitem{2}{5.B (4.42602)}
\treeitem{3}{5.B Composting (4.41264)}
\treeitem{4}{5.B Composting CH$_4$ (2.58076)}
\treeitem{4}{5.B Composting N$_2$O (1.83188)}
\treeitem{3}{5.B.2 Anaerobic Digestion at Biogas Facilities CH$_4$ (0.0133801)}
\treeitem{2}{5.C.1 Incineration of Waste (12.691)}
\treeitem{3}{5.C.1 Incineration of Waste CH$_4$ (0.000140888)}
\treeitem{3}{5.C.1 Incineration of Waste CO$_2$ (12.3575)}
\treeitem{3}{5.C.1 Incineration of Waste N$_2$O (0.333344)}
\treeitem{2}{5.D (42.6765)}
\treeitem{3}{5.D Domestic Wastewater Treatment (34.9967)}
\treeitem{4}{5.D Domestic Wastewater Treatment CH$_4$ (13.5907)}
\treeitem{4}{5.D Domestic Wastewater Treatment N$_2$O (21.406)}
\treeitem{3}{5.D Industrial Wastewater Treatment (7.67978)}
\treeitem{4}{5.D Industrial Wastewater Treatment CH$_4$ (7.21399)}
\treeitem{4}{5.D Industrial Wastewater Treatment N$_2$O (0.465789)}

\end{document}